\documentclass[11pt]{article}
\usepackage{amsxtra}
\usepackage{amssymb, amsmath}
\usepackage{amsthm}
\usepackage{hyperref}
\usepackage{multicol}
\usepackage{latexsym}

\newtheorem{lemma}{Lemma}[section]
\newtheorem{theorem}{Theorem}[section]
\newtheorem{remark}{Remark}[section]

\usepackage[dvips]{graphicx}

\usepackage[T1]{fontenc} 
\usepackage{palatino} 

\begin{document}
\begin{center}
\textbf{\LARGE{Seven Means, Generalized Triangular Discrimination,\\
 and Generating Divergence Measures} }
\end{center}

\bigskip
\begin{center}
\textbf{\large{Inder J. Taneja}}\\
Departamento de Matem\'{a}tica\\
Universidade Federal de Santa Catarina\\
88.040-900 Florian\'{o}polis, SC, Brazil.\\
\textit{e-mail: ijtaneja@gmail.com\\
http://www.mtm.ufsc.br/$\sim $taneja}
\end{center}

\begin{abstract}
  From geometrical point of view, Eve \cite{eve} studied seven means. These means are \textit{Harmonic, Geometric, Arithmetic, Heronian, Contra-harmonic, Root-mean square} and \textit{Centroidal mean}. We have considered for the first time a new measure calling \textit{generalized triangular discrimination}. Inequalities among non-negative differences arising due to seven means and particular cases of \textit{generalized triangular discrimination} are considered. Some new generating measures and their exponential representations are also presented.
\end{abstract}



\section{Seven Geometric Means}

Let $a,\,b>0$ be two positive numbers. Eves \cite{eve} studied the geometrical
interpretation of the following seven means:

\begin{enumerate}
\item Arithmetic mean: \quad $A(a,b)={\left( {a+b} \right)} \mathord{\left/ {\vphantom {{\left( {a+b} \right)} 2}} \right. \kern-\nulldelimiterspace} 2$;
\item Geometric mean: \quad $G(a,b)=\sqrt {ab} $;
\item Harmonic mean: \quad $H(a,b)={2ab} \mathord{\left/ {\vphantom {{2ab} {\left( {a+b} \right)}}} \right. \kern-\nulldelimiterspace} {\left( {a+b} \right)}$;
\item Heronian mean:\quad $N(a,b)={\left( {a+\sqrt {ab} +b} \right)} \mathord{\left/ {\vphantom {{\left( {a+\sqrt {ab} +b} \right)} 3}} \right. \kern-\nulldelimiterspace} 3$;
\item Contra-harmonic mean: \quad $C(a,b)={\left( {a^{2}+b^{2}} \right)} \mathord{\left/ {\vphantom {{\left( {a^{2}+b^{2}} \right)} {\left( {a+b} \right)}}} \right. \kern-\nulldelimiterspace} {\left( {a+b} \right)}$;
\item Root-mean-square: \quad $S(a,b)=\sqrt {{\left( {a^{2}+b^{2}} \right)} \mathord{\left/ {\vphantom {{\left( {a^{2}+b^{2}} \right)} 2}} \right. \kern-\nulldelimiterspace} 2} $;
\item Centroidal mean: \quad $R(a,b)={2\left( {a^{2}+ab+b^{2}} \right)} \mathord{\left/ {\vphantom {{2\left( {a^{2}+ab+b^{2}} \right)} {3\left( {a+b} \right)}}} \right. \kern-\nulldelimiterspace} {3\left( {a+b} \right)}$.
\end{enumerate}

We can easily verify the following inequality having the above seven means:
\begin{equation}
\label{eq1}
H\le G\le N\le A\le R\le S\le C.
\end{equation}

Let us write, $M(a,b)=b\,f_{M} (a/b)$, where $M$ stands for any of the above
seven means, then we have
\begin{equation}
\label{eq2}
f_{H} (x)\le f_{G} (x)
\le f_{N} (x))\le f_{A} (x)\le f_{R} (x)\le f_{S} (x)\le f_{C} (x).
\end{equation}

\smallskip
\noindent where $f_{H} (x)={2x} \mathord{\left/ {\vphantom {{2x} {\left( {x+1}
\right)}}} \right. \kern-\nulldelimiterspace} {\left( {x+1} \right)}$,
$f_{G} (x)=\sqrt x $, $f_{N} (x)={\left( {x+\sqrt x +1} \right)}
\mathord{\left/ {\vphantom {{\left( {x+\sqrt x +1} \right)} 3}} \right.
\kern-\nulldelimiterspace} 3$, $f_{A} (x)={\left( {x+1} \right)}
\mathord{\left/ {\vphantom {{\left( {x+1} \right)} 2}} \right.
\kern-\nulldelimiterspace} 2$, $f_{R} (x)={2\left( {x^{2}+x+1} \right)}
\mathord{\left/ {\vphantom {{2\left( {x^{2}+x+1} \right)} {3\left( {x+1}
\right)}}} \right. \kern-\nulldelimiterspace} {3\left( {x+1} \right)}$,
$f_{S} (x)=\sqrt {{\left( {x^{2}+1} \right)} \mathord{\left/ {\vphantom
{{\left( {x^{2}+1} \right)} 2}} \right. \kern-\nulldelimiterspace} 2} $ \,and
 $f_{C} (x)={\left( {x^{2}+1} \right)} \mathord{\left/ {\vphantom {{\left(
{x^{2}+1} \right)} {\left( {x+1} \right)}}} \right.
\kern-\nulldelimiterspace} {\left( {x+1} \right)}$, $\forall x>0$, $x\ne 1$.
In all these cases, we have equality sign iff $x=1$, i.e., $f_{(\cdot )}
(1)=1$.

\subsection{Inequalities among Differences of Means}

For simplicity, let us write
\begin{equation}
\label{eq3}
D_{AB} =b\,f_{AB} (a,b),
\end{equation}

\smallskip
\noindent
where $f_{UV} (x)=f_{U} (x)-f_{V} (x)$, with $U\ge V$. Thus, according to
(\ref{eq3}), the inequality (\ref{eq1}) admits 21 non-negative differences. These
differences satisfy some simple inequalities given by the following pyramid:
\[
D_{GH};
\]
\[
D_{NG} \le D_{NH} ;
\]
\[
D_{AN} \le D_{AG} \le D_{AH} ;
\]
\[
D_{RA} \le D_{RN} \le D_{RG} \le D_{RH} ;
\]
\[
D_{SR} \le D_{SA} \le D_{SN} \le D_{SG} \le D_{SH} ;
\]
\[
D_{CS} \le D_{CR} \le D_{CA} \le D_{CN} \le D_{CG} \le D_{CH},
\]

\bigskip
\noindent where, for example, $D_{GH} :=G-H$, $D_{NG} :=N-G$, etc. After
simplifications, we have the following equalities among some of these
measures:

\begin{enumerate}
\item $3D_{CR} =2 D_{AH} =2D_{CA} =D_{CH} =6D_{RA} =\textstyle{3 \over 2}D_{RH} :=\Delta ;$
\item $3D_{AN} =D_{AG} = \textstyle{3 \over 2} D_{NG} :=h;$
\item $D_{CG} =3 D_{RN}.$
\end{enumerate}

The measures $\Delta (a,b)$ and $h(a,b)$ are the well know \textit{triangular} and \textit{Hellingar's discriminations} \cite{hel} given by $\Delta (a,b)={(a-b)^{2}} \mathord{\left/ {\vphantom {{(a-b)^{2}} {(a-b)}}}
\right. \kern-\nulldelimiterspace} {(a-b)}$ and $h(a,b)=\textstyle{1 \over
2}\left( {\sqrt a -\sqrt b } \right)^{2}$ respectively. Not all the measures
appearing in the above pyramid are convex in the the pair $(a,b)\in
{\rm R}_{+}^{2} $. Recently, the author \cite{tan6} proved the
following theorem for the convex measures.

\begin{theorem} The following inequalities hold:
\begin{equation}
\label{eq4}
 D_{SA} \le \left\{ {\begin{array}{l}
 \textstyle{3 \over 4}D_{SN} \\
 \textstyle{1 \over 3}D_{SH} \le \textstyle{1 \over 4}\Delta \\
 \end{array}} \right\}\le \left\{ {\begin{array}{l}
 \textstyle{3 \over 7}D_{CN} \le \left\{ {\begin{array}{l}
 D_{CS} \\
 \textstyle{1 \over 3}D_{CG} \le \textstyle{3 \over 5}D_{RG} \\
 \end{array}} \right. \\
 \textstyle{1 \over 2}D_{SG} \le \textstyle{3 \over 5}D_{RG} \\
 \end{array}} \right\} \le h.
 \end{equation}
\end{theorem}

The proof of the above theorem is based on the following two lemmas \cite{tan2,tan4}.

\begin{lemma} Let $f:I\subset {\rm R}_{+} \to {\rm R}$ be a convex and
differentiable function satisfying $f(1)=0$. Consider a function
\[
\varphi_{f} (a,b)=af\left( {\frac{b}{a}} \right),
\quad
a,b>0,
\]
then the function $\varphi_{f} (a,b)$ is convex in ${\rm R}_{+}^{2} $.
Additionally, if $f^{\prime }(1)=0$, then the following inequality hold:
\[
0\le \varphi_{f} (a,b)\le \left( {\frac{b-a}{a}} \right)\varphi_{{f}'}
(a,b).
\]
\end{lemma}

\begin{lemma} Let $f_{1} ,f_{2} :I\subset {\rm R}_{+} \to {\rm R}$ be
two convex functions satisfying the assumptions:

(i) $f_{1} (1)=f_{1}^{\prime }(1)=0$, $f_{2} (1)=f_{2}^{\prime }(1)=0$;

(ii) $f_{1} $ and $f_{2} $ are twice differentiable in ${\rm R}_{+} $;

(iii) there exists the real constants $\alpha ,\beta $ such that $0\le
\alpha <\beta $ and
\[
\alpha \le \frac{f_{1}^{\prime \prime }(x)}{f_{2}^{\prime \prime }(x)}\le
\beta,
\quad
f_{2}^{\prime \prime }(x)>0,
\]
for all $x>0$ then we have the inequalities:
\[
\alpha \mbox{\thinspace }\varphi_{f_{2} } (a,b)\le \varphi_{f_{1} }
(a,b)\le \beta \mbox{\thinspace }\varphi_{f_{2} } (a,b),
\]
for all $a,b\in (0,\infty )$, where the function $\phi_{(\cdot )} (a,b)$ is
as defined in Lemma 1.2.
\end{lemma}

\subsection{Generalized Triangular Discrimination}

For all $a,\,b>0$, let consider the following measures
\begin{equation}
\label{eq5}
L_{t} (a,b)=\frac{\left( {a-b} \right)^{2}\left( {a+b}
\right)^{t}}{2^{t}\left( {\sqrt {ab} } \right)^{t+1}},\quad t\in {\rm Z}
\end{equation}
In particular, we have
\begin{align}
& L_{-1} (a,b)=2\,\Delta (a,b)\notag\\
& L_{0} (a,b)=K(a,b)=\frac{\left( {a-b} \right)^{2}}{\sqrt {ab} },\notag\\
& L_{1} (a,b)=\frac{1}{2}\Psi (a,b)=\frac{\left( {a-b} \right)^{2}\left( {a+b}
\right)}{2ab},\notag\\
& L_{2} (a,b)=\frac{1}{2}F(a,b)=\frac{\left( {a^{2}-b^{2}}
\right)^{2}}{4\left( {ab} \right)^{3/2}}\notag
\intertext{and}
& L_{3} (a,b)=\frac{1}{8}L(a,b)=\frac{\left( {a-b} \right)^{2}\left( {a+b}
\right)^{3}}{8\left( {ab} \right)^{2}}.\notag
\end{align}

From above, we observe that the expression (\ref{eq5}) contains some well-known
measures such as $K(a,b)$ is due to Jain and Srivastava \cite{jas},
$F(a,b)$ is due to Kumar and Johnson \cite{kuj}, $\Psi (a,b)$ is
\textbf{symmetric $\chi^{2}-$measure} \cite{tan2}. While, $L(a,b)$ is
considered here for the first time. More studied on related measures can be seen in
\cite{tan1, tan3, tan5, tak}.

\bigskip
\textbf{Convexity:} Let us prove now the convexity of the measure (\ref{eq5}). We
can write $L_{t} (a,b)=b{\kern 1pt}f(a/b)$, $t\in {\rm Z}$, where
\[
f_{L_{t} } (x)=\frac{\left( {x-1} \right)^{2}\left( {x+1}
\right)^{t}}{2^{t}\left( {\sqrt x } \right)^{t+1}}.
\]
The second order derivative of the function $f_{L_{t} } (x)$ is given by
\[
{f}''_{L_{t} } (x)=\frac{\left( {x+1} \right)^{t-2}}{2^{t+2}x^{2}\left(
{\sqrt x } \right)^{t+1}}\times A_{7} (x,t),
\]
where
\begin{align}
A_{7} (x,t)& =\left( {t+1} \right)\left( {t+3} \right)\left( {x^{4}+1}
\right)+\notag\\
\label{eq6}
& \hspace{10pt} +4x\left( {x^{2}+1} \right)\left( {2-t} \right)\left( {t+1}
\right)+2x^{2}\left( {3t-5} \right)\left( {t-1} \right).
\end{align}

From (\ref{eq6}) we observe that we are unable to find unique value of $t$, when
the function is positive. But for at least $t\in [-1,2]-(1,\textstyle{5
\over 3})$, $x>0$, $x\ne 1$, we have ${f}''_{M_{t} } (x)\ge 0$. Also, we
have $f_{L_{t} } (1)=0$. Thus according to Lemma 1.1, the measure $L_{t}
(a,b)$ is convex for all $(a,b)\in {\rm R}_{+}^{2} $,
$t=-1,0,1\mbox{\thinspace and\thinspace }2$. Testing individually for fix
$t\in {\rm N}$, we can check the convexity for other measures too, for
example for $t=3$, $L_{3} (a,b)$ is convex.

\bigskip
\textbf{Monotonicity: }Calculating the first order derivative of the
function $f_{L_{t} } (x)$ with respect to $t$, we have
\[
\frac{d\left( {f_{L_{t} } (x)} \right)}{dt}=\frac{\left( {x-1}
\right)^{2}\left( {x+1} \right)^{t}}{\left( {2\sqrt x } \right)^{t+1}}\ln
\left( {\frac{\left( {x+1} \right)^{2}}{4x}} \right).
\]
We can easily check that for all $x>0,\,x\ne 1$, ${d\left( {f_{L_{t} } (x)}
\right)} \mathord{\left/ {\vphantom {{d\left( {f_{L_{t} } (x)} \right)}
{dt}}} \right. \kern-\nulldelimiterspace} {dt}>0$. This proves that the
function$f_{L_{t} } (x)$ is decreasing with respect to $t$. In view of this
we have
\begin{equation}
\label{eq7}
\textstyle{1 \over 4}\Delta \le h\le \textstyle{1 \over 8}K\le \textstyle{1
\over {16}}\Psi \le \textstyle{1 \over {16}}F\le \textstyle{1 \over {64}}L.
\end{equation}

Also we know that $h(a,b)\le \textstyle{1 \over 8}K(a,b)$. Thus combining
(\ref{eq4}) and (\ref{eq7}), we have
\[
 D_{SA} \le \left\{ {\begin{array}{l}
 \textstyle{3 \over 4}D_{SN} \\
 \textstyle{1 \over 3}D_{SH} \le \textstyle{1 \over 4}\Delta \\
 \end{array}} \right\}\le \left\{ {\begin{array}{l}
 \textstyle{3 \over 7}D_{CN} \le \left\{ {\begin{array}{l}
 D_{CS} \\
 \textstyle{1 \over 3}D_{CG} \le \textstyle{3 \over 5}D_{RG} \\
 \end{array}} \right. \\
 \textstyle{1 \over 2}D_{SG} \le \textstyle{3 \over 5}D_{RG} \\
 \end{array}} \right\} \le
 \]
\begin{equation}
\label{eq8}
\le h\le \textstyle{1 \over 8}K\le
\textstyle{1 \over {16}}\Psi \le \textstyle{1 \over {16}}F\le \textstyle{1
\over {64}}L.
\end{equation}

As a part of (\ref{eq8}), let us consider the following inequalities:
\begin{equation}
\label{eq9}
2\Delta \le \textstyle{{24} \over 7}D_{CN} \le \textstyle{8 \over 3}D_{CG}
\le \textstyle{{24} \over 5}D_{RG} \le 8h\le K\le \textstyle{1 \over 2}\Psi
\le \textstyle{1 \over 2}F\le \textstyle{1 \over 8}L.
\end{equation}

\begin{remark} Here we have considered only the positive numbers
$a,\,b>0$. All the results remains valid if we consider the probability
distributions $P,Q\in \Gamma_{n} $, where
\[
\Gamma_{n} =\left\{ {P=(p_{1} ,p_{2} ,...,p_{n} )\left| {p_{i}
>0,\sum\limits_{i=1}^n {p_{i} =1} } \right.} \right\},
\, n\ge 2.
\]
\begin{itemize}
\item [(i)] In terms of probability distributions the measure $L_{t} (a,b)$ is written
as
\begin{equation}
\label{eq9a}
L_{t} (P\vert \vert Q)=\sum\limits_{i=1}^n {\frac{\left( {p_{i} -q_{i} }
\right)^{2}\left( {p_{i} +q_{i} } \right)^{t}}{2^{t}\left( {\sqrt {p_{i}
q_{i} } } \right)^{t+1}}} ,\;(P,Q)\in \Gamma_{n} \times \Gamma_{n},\, t\in {\rm Z}
\end{equation}
\item[(ii)] Topsoe \cite{top} considered a different type of generalized triangular discrimination
\begin{equation}
\label{eq9b}
\Delta_{t} (P\vert \vert Q)=\sum\limits_{i=1}^n {\frac{\left( {p_{i} -q_{i}
} \right)^{2t}}{\left( {p_{i} +q_{i} } \right)^{2t-1}}} ,\;(P,Q)\in \Gamma
_{n} \times \Gamma_{n} ,\;t\in {\rm N}_{+}
\end{equation}
\end{itemize}
\end{remark}

In this paper our aim is to study further inequalities by considering the
possible nonnegative differences arising due to (\ref{eq9}).

\section{New Inequalities}

In this section we shall bring inequalities in different stages. In the
first stage the measures considered are the nonnegative differences arsing
due to (\ref{eq9}). This we have done many times until we left with one measure
in the final stage.

\subsection{First Stage}

For simplicity, let us write the expression (\ref{eq9}) as
\begin{equation}
\label{eq10}
W_{1} \le W_{2} \le W_{3} \le W_{4} \le W_{5} \le W_{6} \le W_{7} \le W_{8}
\le W_{9} ,
\end{equation}
where for example $W_{1} =2\Delta $, $W_{2} =\textstyle{{24} \over 7}D_{CN}
$, $W_{9} =\textstyle{8 \over 3}D_{CG} $, etc. We can write
\begin{equation}
\label{eq11}
W_{t} (a,b):=b\,f\left( {\frac{a}{b}} \right),
\quad
t=1,2,...,9,
\end{equation}

\noindent where

\begin{align}
& f_{W_{1} } (x)=2f_{\Delta } (x)=\frac{2\left( {x-1} \right)^{2}}{x+1},\notag\\
& f_{W_{2} } (x)=\frac{24}{7}f_{CN} (x)=\frac{8\left( {\sqrt x -1}
\right)^{2}\left( {2x+3\sqrt x +2} \right)}{7\left( {x+1} \right)},\notag\\
& f_{W_{3} } (x)=\frac{8}{3}f_{CG} (x)=\frac{8\left( {\sqrt x -1}
\right)^{2}\left( {x+\sqrt x +1} \right)}{3\left( {x+1} \right)},\notag
\end{align}
\begin{align}
& f_{W_{4} } (x)=\frac{24}{5}f_{RG} (x)=\frac{8\left( {\sqrt x -1}
\right)^{2}\left( {2x+\sqrt x +2} \right)}{5\left( {x+1} \right)},\notag\\
& f_{W_{5} } (x)=8f_{h} (x)=4\left( {\sqrt x -1} \right)^{2},\notag\\
& f_{W_{6} } (x)=f_{K} (x)=\frac{\left( {x-1} \right)^{2}}{\sqrt x },\notag\\
& f_{W_{7} } (x)=\frac{1}{2}f_{\Psi } (x)=\frac{\left( {x-1}
\right)^{2}\left( {x+1} \right)}{2x},\notag\\
& f_{W_{8} } (x)=\frac{1}{4}f_{F} (x)=\frac{\left( {x^{2}-1}
\right)^{2}}{2x^{3/2}},\notag
\intertext{and}
& f_{W_{9} } (x)=\frac{1}{8}f_{L} (x)=\frac{\left( {x-1} \right)^{2}\left(
{x+1} \right)^{3}}{8x^{2}}.\notag
\end{align}

Calculating the second order derivative of above functions we have
\begin{align}
& {f}''_{W_{1} } (x)=\frac{16}{\left( {x+1} \right)^{3}},\notag\\
& {f}''_{W_{2} } (x)=\frac{2\left[ {\left( {x+1} \right)^{3}+48x^{3/2}}
\right]}{7x^{3/2}\left( {x+1} \right)^{3}},\notag\\
& {f}''_{W_{3} } (x)=\frac{2\left[ {\left( {x+1} \right)^{3}+16x^{3/2}}
\right]}{3x^{3/2}\left( {x+1} \right)^{3}},\notag\\
& {f}''_{W_{4} } (x)=\frac{2\left[ {3\left( {x+1} \right)^{3}+16x^{3/2}}
\right]}{5x^{3/2}\left( {x+1} \right)^{3}},\notag\\
& {f}''_{W_{5} } (x)=\frac{2}{x^{3/2}},\notag\\
& {f}''_{W_{6} } (x)=\frac{3x^{2}+2x+3}{4x^{5/2}};\notag\\
& {f}''_{W_{7} } (x)=\frac{x^{3}+1}{x^{3}},\notag\\
& {f}''_{W_{8} } (x)=\frac{14x^{4}+2x^{2}+15}{16x^{7/2}},\notag
\intertext{and}
& {f}''_{W_{9} } (x)=\frac{\left( {x+1} \right)\left[ {2\left( {x^{4}+1}
\right)+\left( {x^{2}+1} \right)\left( {x-1} \right)^{2}} \right]}{4x^{4}}.\notag
\end{align}

The inequalities (\ref{eq11}) again admits 45 nonnegative differences. These
differences satisfies some natural inequalities given in a \textbf{pyramid}
below:
\newpage
\[
D_{W_{2} W_{1} }^{1} ;
\]
\[
D_{W_{3} W_{2} }^{2} \le D_{W_{3} W_{1} }^{3} ;
\]
\[
D_{W_{4} W_{3} }^{4} \le D_{W_{4} W_{2} }^{5} \le D_{W_{4} W_{1} }^{6} ;
\]
\[
D_{W_{5} W_{4} }^{7} \le D_{W_{5} W_{3} }^{8} \le D_{W_{5} W_{2} }^{9} \le D_{W_{5} W_{1} }^{10} ;
\]
\[
 D_{W_{6} W_{5} }^{11} \le D_{W_{6} W_{4} }^{12} \le D_{W_{6} W_{3} }^{13} \le D_{W_{6} W_{2} }^{14} \le D_{W_{6} W_{1} }^{15} ;
 \]
\[
 D_{W_{7} W_{6} }^{16} \le D_{W_{7} W_{5} }^{17} \le D_{W_{7} W_{4} }^{18} \le D_{W_{7} W_{3} }^{19} \le D_{W_{7} W_{2} }^{20} \le D_{W_{7} W_{1} }^{21} ;
 \]
\[
 D_{W_{8} W_{7} }^{22} \le D_{W_{8} W_{6} }^{23} \le D_{W_{8} W_{5} }^{24} \le D_{W_{8} W_{4} }^{25} \le D_{W_{8} W_{3} }^{26} \le D_{W_{8} W_{2} }^{27} \le D_{W_{8} W_{1} }^{28} ;
 \]
\[
 D_{W_{9} W_{8} }^{29} \le D_{W_{9} W_{7} }^{30} \le D_{W_{9} W_{6} }^{31} \le D_{W_{9} W_{5} }^{32} \le D_{W_{9} W_{4} }^{33} \le D_{W_{9} W_{3} }^{34} \le D_{W_{9} W_{2} }^{35} \le D_{W_{9} W_{1} }^{36}.
 \]

\bigskip
\noindent where $D_{W_{2} W_{1} }^{1} :=W_{2} -W_{1} $, $D_{W_{7} W_{6} }^{16} :=W_{7}
-W_{6} $, etc. After simplifications, we have equalities among first four
lines of the \textbf{pyramid: }
\begin{align}
& \textstyle{7 \over 2}D_{W_{2} W_{1} }^{1} =\textstyle{{21} \over 8}D_{W_{3}
W_{2} }^{2} =\textstyle{3 \over 2}D_{W_{3} W_{1} }^{3} =\textstyle{{15}
\over 8}D_{W_{4} W_{3} }^{4} =\textstyle{{35} \over {32}}D_{W_{4} W_{2}
}^{5} = \notag\\
\label{eq11a}
& \hspace{15pt} =\textstyle{5 \over 6}D_{W_{4} W_{1} }^{6} =\textstyle{5 \over 4}D_{W_{5}
W_{4} }^{7} =\textstyle{3 \over 4}D_{W_{5} W_{3} }^{8} =\textstyle{7 \over
{12}}D_{W_{5} W_{2} }^{9} =\textstyle{1 \over 2}D_{W_{5} W_{1} }^{10}
=\frac{\left( {\sqrt a -\sqrt b } \right)^{4}}{a+b}.
\end{align}

\bigskip
In view of above equalities we are left only with 27 nonnegative convex
measures and these are connected with each other by inequalities given in
the theorem below.

\begin{theorem} The following sequences of inequalities hold:
\begin{align}
& D_{W_{2} W_{1} }^{1} \le \textstyle{1 \over {14}}D_{W_{6} W_{1} }^{15} \le
\textstyle{1 \over {13}}D_{W_{6} W_{2} }^{14} \le D_{W_{6} W_{3} }^{13} \le
D_{W_{6} W_{4} }^{12} \le D_{W_{6} W_{5} }^{11} \le\notag\\
& \hspace{10pt} \le D_{W_{7} W_{1} }^{21} \le D_{W_{7} W_{2} }^{20} \le D_{W_{7} W_{3}
}^{19} \le D_{W_{7} W_{4} }^{18} \le D_{W_{7} W_{5} }^{17} \le D_{W_{7}
W_{6} }^{16} \le\notag\\
& \hspace{20pt} \le D_{W_{8} W_{1} }^{28} \le D_{W_{8} W_{2} }^{27} \le D_{W_{8} W_{3}
}^{26} \le D_{W_{8} W_{4} }^{25} \le D_{W_{8} W_{5} }^{24} \le D_{W_{8}
W_{6} }^{23} \le \left\{ {\begin{array}{l}
 D_{W_{8} W_{7} }^{22} \\
 D_{W_{9} W_{1} }^{36} \\
 \end{array}} \right\}\le \notag\\
\label{eq12}
& \hspace{30pt}\le D_{W_{9} W_{2} }^{35} \le D_{W_{9} W_{3} }^{34} \le D_{W_{9} W_{4}
}^{33} \le D_{W_{9} W_{5} }^{32} \le D_{W_{9} W_{6} }^{31} \le D_{W_{9}
W_{7} }^{30} \le D_{W_{9} W_{8} }^{29}.
\end{align}
\end{theorem}

\begin{proof} We shall prove the above theorem by parts.

\begin{enumerate}
\item \textbf{For }$\bf{D_{W_{2} W_{1} }^{1} \le \textstyle{1 \over {14}}D_{W_{6} W_{1} }^{15}} $\textbf{: }We shall apply two approach to prove this result.

\textbf{1}$^{\mathrm{\mathbf{st}}}$\textbf{ Approach:} Let us consider a
function
\[
g_{W_{2} W_{1} \mathunderscore W_{6} W_{1} } (x)=\frac{{f}''_{W_{2} W_{1} }
(x)}{{f}''_{W_{6} W_{1} } (x)}
\quad
=\frac{{f}''_{W_{2} } (x)-{f}''_{W_{1} } (x)}{{f}''_{W_{6} }
(x)-{f}''_{W_{1} } (x)}
\]
After simplifications, we have
\[
g_{W_{2} W_{1} \mathunderscore W_{6} W_{1} } (x)=\frac{8x\left(
{x^{2}+2x^{3/2}+6x+2\sqrt x +1} \right)}{7\left( {\begin{array}{l}
 3x^{4}+6x^{7/2}+20x^{3}+34x^{5/2}+ \\
 +66x^{2}+34x^{3/2}+20x+6\sqrt x +3 \\
 \end{array}} \right)}
\]
\begin{align}
& {g}'_{W_{2} W_{1} \mathunderscore W_{6} W_{1} } (x) =\notag\\
& \hspace{10pt} =-\frac{48\left( {x-1}
\right)\left( {x+1} \right)^{2}\left( {\begin{array}{l}
 x^{3}+4x^{5/2}+15x^{2}+ \\
 +20x^{3/2}+15x+4\sqrt x +1 \\
 \end{array}} \right)}{7\left( {\begin{array}{l}
 3x^{4}+6x^{7/2}+20x^{3}+34x^{5/2}+ \\
 +66x^{2}+34x^{3/2}+20x+6\sqrt x +3 \\
 \end{array}} \right)}
\begin{cases}
 {>0} & {x<1} \\
 {<0} & {x>1} \\
\end{cases}\notag
\end{align}
and
\[
\beta_{W_{2} W_{1} \mathunderscore W_{6} W_{1} } =\mathop {\sup }\limits_{x\in
(0,\infty )} g_{W_{2} W_{1} \mathunderscore W_{6} W_{1} } (x)=g_{W_{2} W_{1}
\mathunderscore W_{6} W_{1} } (1)=\frac{1}{14}.
\]
By the application Lemma 1.2 we get the required result.

\bigskip
\textbf{2}$^{\mathrm{\mathbf{nd}}}$\textbf{ Approach:} We shall use an
alternative approach to prove the above result. We know that $\beta_{W_{2}
W_{1} \mathunderscore W_{6} W_{1} } =g_{W_{2} W_{1} \mathunderscore W_{6}
W_{1} } (1)={{f}''_{W_{2} W_{1} } (1)} \mathord{\left/ {\vphantom
{{{f}''_{W_{2} W_{1} } (1)} {{f}''_{W_{6} W_{1} } (1)}}} \right.
\kern-\nulldelimiterspace} {{f}''_{W_{6} W_{1} } (1)}=\textstyle{1 \over
{14}}$. In order to prove the result we need to show that $\textstyle{1
\over {14}}D_{W_{6} W_{1} }^{15} -D_{W_{2} W_{1} }^{1} \ge 0$. By
considering the difference $\textstyle{1 \over {14}}D_{W_{6} W_{1} }^{15}
-D_{W_{2} W_{1} }^{1} $, we have
\[
\frac{1}{14}D_{W_{6} W_{1} }^{15} -D_{W_{2} W_{1} }^{1}
=\frac{1}{14}\left( {W_{6} +13W_{1} -14W_{2} } \right)
=\frac{1}{14}V_{1} :=b{\kern 1pt}f_{V_{1} } \left( {\frac{a}{b}} \right),
\]
where
\begin{equation}
\label{eq13}
f_{V_{1} } (x)=\frac{\left( {\sqrt x -1} \right)^{6}}{\sqrt x \left( {x+1}
\right)}>0,
\quad
\forall x>0,\,x\ne 1.
\end{equation}
Since $V_{1} (a,b)\ge 0$, we get the required result.

\bigskip
\textbf{Note: }\textit{In two approaches applied above, we observe that the second one is easier. Moreover, in some cases we are unable to deduce the results by applications of first approach. In proving other parts, we shall only apply the second one. Without specifying, we shall frequently use the second derivatives }${f}''_{W_{t} } (x), t=1,2,...,9$\textit{ written above.}

\item \textbf{For }$\bf{D_{W_{6} W_{1} }^{15} \le \textstyle{{14} \over {13}}D_{W_{6} W_{2} }^{14}} $\textbf{: }Let us consider a function $g_{W_{6} W_{1} \mathunderscore W_{6} W_{2} } (x)={{f}''_{W_{6} W_{1} } (x)} \mathord{\left/ {\vphantom {{{f}''_{W_{6} W_{1} } (x)} {{f}''_{W_{6} W_{2} } (x)}}} \right. \kern-\nulldelimiterspace} {{f}''_{W_{6} W_{2} } (x)}$. After simplifications, we have
\[
g_{W_{6} W_{1} \mathunderscore W_{6} W_{2} } (x)=\frac{7\left(
{\begin{array}{l}
 3x^{4}+6x^{7/2}+20x^{3}+34x^{5/2}+ \\
 +66x^{2}+34x^{3/2}+20x+6\sqrt x +3 \\
 \end{array}} \right)}{3\left( {\begin{array}{l}
 7x^{4}+14x^{7/2}+44x^{3}+74x^{5/2}+ \\
 +138x^{2}+74x^{3/2}+44x+14\sqrt x +7 \\
 \end{array}} \right)},
\]
\[
\beta_{W_{6} W_{1} \mathunderscore W_{6} W_{2} } =g_{W_{6} W_{1}
\mathunderscore W_{6} W_{2} } (1)=\frac{14}{13}
\]
and
\[
\frac{14}{13}D_{W_{6} W_{2} }^{14} -D_{W_{6} W_{1} }^{15}
=\frac{1}{13}\left( {W_{6} +13W_{1} -14W_{2} } \right)
=\frac{1}{13}V_{1} .
\]

\item \textbf{For }$\bf{D_{W_{6} W_{2} }^{14} \le \textstyle{{39} \over {35}}D_{W_{6} W_{3} }^{13}} $\textbf{: }Let us consider a function $g_{W_{6} W_{2} \mathunderscore W_{6} W_{3} } (x)={{f}''_{W_{6} W_{2} } (x)} \mathord{\left/ {\vphantom {{{f}''_{W_{6} W_{2} } (x)} {{f}''_{W_{6} W_{3} } (x)}}} \right. \kern-\nulldelimiterspace} {{f}''_{W_{6} W_{3} } (x)}$. After simplifications, we have
\[
g_{W_{6} W_{2} \mathunderscore W_{6} W_{3} } (x)=\frac{9}{7}\frac{\left(
{\begin{array}{l}
 7x^{4}+14x^{7/2}+44x^{3}+74x^{5/2}+ \\
 +138x^{2}+74x^{3/2}+44x+14\sqrt x +7 \\
 \end{array}} \right)}{\left( {\begin{array}{l}
 9x^{4}+18x^{7/2}+52x^{3}+86x^{5/2}+ \\
 +150x^{2}+86x^{3/2}+52x+18\sqrt x +9 \\
 \end{array}} \right)},
\]
\[
\beta_{W_{6} W_{2} \mathunderscore W_{6} W_{3} } =g_{W_{6} W_{2}
\mathunderscore W_{6} W_{3} } (1)=\frac{39}{35}
\]
and
\[
\frac{39}{35}D_{W_{6} W_{3} }^{13} -D_{W_{6} W_{2} }^{14}
=\frac{1}{35}\left( {4W_{6} +35W_{2} -39W_{2} } \right)
=\frac{4}{35}V_{1}.
\]

\item \textbf{For }$D_{W_{6} W_{3} }^{13} \le \textstyle{{25} \over {21}}D_{W_{6} W_{4} }^{12} $\textbf{: }Let us consider a function $g_{W_{6} W_{3} \mathunderscore W_{6} W_{4} } (x)={{f}''_{W_{6} W_{3} } (x)} \mathord{\left/ {\vphantom {{{f}''_{W_{6} W_{3} } (x)} {{f}''_{W_{6} W_{4} } (x)}}} \right. \kern-\nulldelimiterspace} {{f}''_{W_{6} W_{4} } (x)}$. After simplifications, we have
\[
g_{W_{6} W_{3} \mathunderscore W_{6} W_{4} } (x)=\frac{5\left(
{\begin{array}{l}
 9x^{4}+18x^{7/2}+52x^{3}+86x^{5/2}+ \\
 +150x^{2}+86x^{3/2}+52x+18\sqrt x +9 \\
 \end{array}} \right)}{3\left( {\begin{array}{l}
 15x^{4}+30x^{7/2}+76x^{3}+122x^{5/2}+ \\
 +186x^{2}+122x^{3/2}+76x+30\sqrt x +15 \\
 \end{array}} \right)},
\]
\[
\beta_{W_{6} W_{3} \mathunderscore W_{6} W_{4} } =g_{W_{6} W_{3}
\mathunderscore W_{6} W_{4} } (1)=\frac{25}{21}
\]
and
\[
\frac{25}{21}D_{W_{6} W_{4} }^{12} -D_{W_{6} W_{3} }^{13}
=\frac{1}{21}\left( {4W_{6} +21W_{3} -25W_{4} } \right)
=\frac{4}{21}V_{1} .
\]

\item \textbf{For }$\bf{D_{W_{6} W_{4} }^{12} \le \textstyle{7 \over 5}D_{W_{6} W_{5} }^{11}} $\textbf{: }Let us consider a function $g_{W_{6} W_{4} \mathunderscore W_{6} W_{5} } (x)={{f}''_{W_{6} W_{4} } (x)} \mathord{\left/ {\vphantom {{{f}''_{W_{6} W_{4} } (x)} {{f}''_{W_{6} W_{5} } (x)}}} \right. \kern-\nulldelimiterspace} {{f}''_{W_{6} W_{5} } (x)}$. After simplifications, we have
\[
g_{W_{6} W_{4} \mathunderscore W_{6} W_{5} } (x)=\frac{\left(
{\begin{array}{l}
 15x^{4}+30x^{7/2}+76x^{3}+122x^{5/2}+ \\
 +186x^{2}+122x^{3/2}+76x+30\sqrt x +15 \\
 \end{array}} \right)}{15\left( {\sqrt x +1} \right)^{2}\left( {x+1}
\right)^{3}},
\]
\[
\beta_{W_{6} W_{4} \mathunderscore W_{6} W_{5} } =g_{W_{6} W_{4}
\mathunderscore W_{6} W_{5} } (1)=\frac{7}{5}.
\]
and
\[
\frac{5}{3}D_{W_{4} W_{2} }^{5} -D_{W_{5} W_{2} }^{9}
=\frac{1}{5}\left( {2W_{6} +5W_{4} -7W_{5} } \right)
=\frac{2}{5}V_{1} \quad .
\]

\item \textbf{For }$\bf{D_{W_{6} W_{5} }^{11} \le \textstyle{1 \over 4}D_{W_{7} W_{1} }^{21}} $\textbf{: }Let us consider a function $g_{W_{6} W_{5} \mathunderscore W_{7} W_{1} } (x)={{f}''_{W_{6} W_{5} } (x)} \mathord{\left/ {\vphantom {{{f}''_{W_{6} W_{5} } (x)} {{f}''_{W_{7} W_{1} } (x)}}} \right. \kern-\nulldelimiterspace} {{f}''_{W_{7} W_{1} } (x)}$. After simplifications, we have
\[
g_{W_{6} W_{5} \mathunderscore W_{7} W} (x)=\frac{3\sqrt x \left( {x+1}
\right)^{3}}{4\left( {x^{4}+5x^{3}+12x^{2}+5x+1} \right)},
\]
\[
\beta_{W_{6} W_{5} \mathunderscore W_{7} W_{1} } =g_{W_{6} W_{5}
\mathunderscore W_{7} W_{1} } (1)=\frac{1}{4}.
\]
and
\[
\frac{1}{4}D_{W_{7} W_{1} }^{21} -D_{W_{6} W_{5} }^{11}
=\frac{1}{4}\left( {W_{7} +4W_{5} -W_{1} -4W_{6} } \right)
=\frac{1}{8}V_{2} :=\frac{1}{8}b{\kern 1pt}f_{V_{2} } \left( {\frac{a}{b}}
\right),
\]
where
\begin{equation}
\label{eq14}
f_{V_{2} } (x)==\frac{\left( {\sqrt x -1} \right)^{8}}{x\left( {x+1}
\right)}>0,
\quad
\forall x>0,\,x\ne 1.
\end{equation}

\item \textbf{For }$\bf{D_{W_{7} W_{1} }^{21} \le \textstyle{{28} \over {27}}D_{W_{7} W_{2} }^{20}} $\textbf{: }Let us consider a function $g_{W_{7} W_{1} \mathunderscore W_{7} W_{2} } (x)={{f}''_{W_{7} W_{1} } (x)} \mathord{\left/ {\vphantom {{{f}''_{W_{7} W_{1} } (x)} {{f}''_{W_{7} W_{2} } (x)}}} \right. \kern-\nulldelimiterspace} {{f}''_{W_{7} W_{2} } (x)}$. After simplifications, we have
\[
g_{W_{7} W_{1} \mathunderscore W_{7} W_{2} } (x)=\frac{7\left( {\sqrt x +1}
\right)^{2}\left( {x^{4}+5x^{3}+12x^{2}+5x+1} \right)}{\left(
{\begin{array}{l}
 7x^{5}+14x^{9/2}+42x^{4}+68x^{7/2}+115x^{3}+ \\
 +156x^{5/2}+115x^{2}+68x^{3/2}+42x+14\sqrt x +7 \\
 \end{array}} \right)}
\]
\[
\beta_{W_{7} W_{1} \mathunderscore W_{7} W_{2} } =g_{W_{7} W_{1}
\mathunderscore W_{7} W_{2} } (1)=\frac{28}{27}.
\]
and
\[
\frac{28}{27}D_{W_{7} W_{2} }^{20} -D_{W_{7} W_{1} }^{21}
=\frac{1}{27}\left( {W_{7} +27W_{1} -28W_{2} } \right)
=\frac{1}{54}V_{3} :=\frac{1}{54}b{\kern 1pt}f_{V_{3} } \left( {\frac{a}{b}}
\right),
\]
where
\begin{equation}
\label{eq15}
f_{V_{3} } (x)=\frac{\left( {x+6\sqrt x +1} \right)\left( {\sqrt x -1}
\right)^{6}}{x\left( {x+1} \right)}>0,
\quad
\forall x>0,\,\,x\ne 1
\end{equation}

\item \textbf{For }$\bf{D_{W_{7} W_{2} }^{20} \le \textstyle{{81} \over {77}}D_{W_{7} W_{3} }^{19} }$\textbf{: }Let us consider a function $g_{W_{7} W_{2} \mathunderscore W_{7} W_{3} } (x)={{f}''_{W_{7} W_{2} } (x)} \mathord{\left/ {\vphantom {{{f}''_{W_{7} W_{2} } (x)} {{f}''_{W_{7} W_{3} } (x)}}} \right. \kern-\nulldelimiterspace} {{f}''_{W_{7} W_{3} } (x)}$. After simplifications, we have
\[
g_{W_{7} W_{2} \mathunderscore W_{7} W_{3} } (x)=\frac{3\left(
{\begin{array}{l}
 7x^{5}+14x^{9/2}+42x^{4}+68x^{7/2}+115x^{3}+ \\
 +156x^{5/2}+115x^{2}+68x^{3/2}+42x+14\sqrt x +7 \\
 \end{array}} \right)}{7\left( {\begin{array}{l}
 3x^{5}+6x^{9/2}+18x^{4}+28x^{7/2}+47x^{3}+ \\
 +60x^{5/2}+47x^{2}+28x^{3/2}+18x+6\sqrt x +3 \\
 \end{array}} \right)},
\]
\[
\beta_{W_{7} W_{2} \mathunderscore W_{7} W_{3} } =g_{W_{7} W_{2}
\mathunderscore W_{7} W_{3} } (1)=\frac{81}{77}.
\]
and
\[
\frac{81}{77}D_{W_{7} W_{3} }^{19} -D_{W_{7} W_{2} }^{20}
=\frac{1}{77}\left( {4W_{7} +77W_{2} -81W_{3} } \right)=\frac{2}{77}V_{3} .
\]

\item \textbf{For }$D_{W_{7} W_{3} }^{19} \le \textstyle{{55} \over {51}}D_{W_{7} W_{4} }^{18} $\textbf{: }Let us consider a function $g_{W_{7} W_{3} \mathunderscore W_{7} W_{4} } (x)={{f}''_{W_{7} W_{3} } (x)} \mathord{\left/ {\vphantom {{{f}''_{W_{7} W_{3} } (x)} {{f}''_{W_{7} W_{4} } (x)}}} \right. \kern-\nulldelimiterspace} {{f}''_{W_{7} W_{4} } (x)}$. After simplifications, we have
\[
g_{W_{7} W_{3} \mathunderscore W_{7} W_{4} } (x)=\frac{5\left(
{\begin{array}{l}
 3x^{5}+6x^{9/2}+18x^{4}+28x^{7/2}+47x^{3}+ \\
 +60x^{5/2}+47x^{2}+28x^{3/2}+18x+6\sqrt x +3 \\
 \end{array}} \right)}{3\left( {\begin{array}{l}
 5x^{5}+10x^{9/2}+30x^{4}+44x^{7/2}+73x^{3}+ \\
 +84x^{5/2}+73x^{2}+44x^{3/2}+30x+10\sqrt x +5 \\
 \end{array}} \right)}
\]
\[
\beta_{W_{7} W_{3} \mathunderscore W_{7} W_{4} } =g_{W_{7} W_{3}
\mathunderscore W_{7} W_{4} } (1)=\frac{55}{51}
\]
and
\[
\frac{55}{51}D_{W_{7} W_{4} }^{18} -D_{W_{7} W_{3} }^{19}
=\frac{1}{51}\left( {4W_{7} +51W_{3} -55W_{4} } \right)
=\frac{2}{51}V_{3} \quad .
\]

\item \textbf{For }$\bf{D_{W_{7} W_{4} }^{18} \le \textstyle{{17} \over {15}}D_{W_{7} W_{5} }^{17}} $\textbf{: }Let us consider a function $g_{W_{7} W_{4} \mathunderscore W_{7} W_{5} } (x)={{f}''_{W_{7} W_{4} } (x)} \mathord{\left/ {\vphantom {{{f}''_{W_{7} W_{4} } (x)} {{f}''_{W_{7} W_{5} } (x)}}} \right. \kern-\nulldelimiterspace} {{f}''_{W_{7} W_{5} } (x)}$. After simplifications, we have
\[
g_{W_{7} W_{4} \mathunderscore W_{7} W_{5} } (x)=\frac{\left(
{\begin{array}{l}
 5x^{5}+10x^{9/2}+30x^{4}+44x^{7/2}+73x^{3}+ \\
 +84x^{5/2}+73x^{2}+44x^{3/2}+30x+10\sqrt x +5 \\
 \end{array}} \right)}{5\left( {x+\sqrt x +1} \right)^{2}\left( {x+1}
\right)^{3}},
\]
\[
\beta_{W_{7} W_{4} \mathunderscore W_{7} W_{5} } =g_{W_{7} W_{4}
\mathunderscore W_{7} W_{5} } (1)=\frac{17}{15}
\]
and
\[
\frac{17}{15}D_{W_{7} W_{5} }^{17} -D_{W_{7} W_{4} }^{18}
=\frac{1}{15}\left( {2W_{7} +15W_{4} -17W_{5} } \right)=\frac{1}{15}V_{3} .
\]

\item \textbf{For }$D_{W_{7} W_{5} }^{17} \le \textstyle{3 \over 2}D_{W_{7} W_{6} }^{16} $\textbf{: }Let us consider a function $g_{W_{7} W_{5} \mathunderscore W_{7} W_{6} } (x)={{f}''_{W_{7} W_{5} } (x)} \mathord{\left/ {\vphantom {{{f}''_{W_{7} W_{5} } (x)} {{f}''_{W_{7} W_{6} } (x)}}} \right. \kern-\nulldelimiterspace} {{f}''_{W_{7} W_{6} } (x)}$. After simplifications, we have
\[
g_{W_{7} W_{5} \mathunderscore W_{7} W_{6} } (x)=\frac{4\left( {x+\sqrt x
+1} \right)^{2}}{4x^{2}+5x^{3/2}+6x+5\sqrt x +4},
\]
\[
\beta_{W_{7} W_{5} \mathunderscore W_{7} W_{6} } =g_{W_{7} W_{5}
\mathunderscore W_{7} W_{6} } (1)=\frac{3}{2}
\]
and
\[
\frac{3}{2}D_{W_{7} W_{6} }^{16} -D_{W_{7} W_{5} }^{17}
=\frac{1}{2}\left( {W_{7} +2W_{5} -3W_{6} } \right)=\frac{1}{4}V_{4}
:=\frac{1}{4}b{\kern 1pt}f_{V_{4} } \left( {\frac{a}{b}} \right),
\]
where
\begin{equation}
\label{eq16}
f_{V_{4} } (x)=\frac{\left( {\sqrt x -1} \right)^{6}}{x}>0,
\quad
\forall x>0,\,\,x\ne 1.
\end{equation}

\item \textbf{For }$\bf{D_{W_{7} W_{6} }^{16} \le \textstyle{1 \over 3}D_{W_{8} W_{1} }^{28}} $\textbf{: }Let us consider a function $g_{W_{7} W_{6} \mathunderscore W_{8} W_{1} } (x)={{f}''_{W_{7} W_{6} } (x)} \mathord{\left/ {\vphantom {{{f}''_{W_{7} W_{6} } (x)} {{f}''_{W_{8} W_{1} } (x)}}} \right. \kern-\nulldelimiterspace} {{f}''_{W_{8} W_{1} } (x)}$. After simplifications, we have
\[
g_{W_{7} W_{6} \mathunderscore W_{8} W_{1} } (x)=\frac{4\sqrt x \left( {x+1}
\right)^{3}\left( {4x^{2}+5x^{3/2}+6x+5\sqrt x +4} \right)}{\left(
{\begin{array}{l}
 15+364x^{5/2}+30\sqrt x +492x^{3}+364x^{7/2}+ \\
 +150x^{9/2}+90x+90x^{5}+30x^{11/2}+15x^{6}+ \\
 +257x^{4}+257x^{2}+150x^{3/2} \\
 \end{array}} \right)},
\]
\[
\beta_{W_{7} W_{6} \mathunderscore W_{8} W_{1} } =g_{W_{7} W_{6}
\mathunderscore W_{8} W_{1} } (1)=\frac{1}{3}
\]
and
\[
\frac{1}{3}D_{W_{8} W_{1} }^{28} -D_{W_{7} W_{6} }^{16}
=\frac{1}{3}\left( {W_{8} +3W_{6} -W_{1} -3W_{7} } \right)
=\frac{1}{12}V_{5} :=b{\kern 1pt}f_{V_{5} } \left( {\frac{a}{b}} \right),
\]
where
\begin{equation}
\label{eq17}
f_{V_{_{5} } } (x)=\frac{\left( {\sqrt x +1} \right)^{2}\left( {\sqrt x -1}
\right)^{8}}{x^{3/2}\left( {x+1} \right)}>0,
\quad
\forall x>0,\,\,x\ne 1.
\end{equation}

\item \textbf{For }$\bf{D_{W_{8} W_{1} }^{28} \le \textstyle{{42} \over {41}}D_{W_{8} W_{2} }^{27} }$\textbf{: }Let us consider a function $g_{W_{8} W_{1} \mathunderscore W_{8} W_{2} } (x)={{f}''_{W_{8} W_{1} } (x)} \mathord{\left/ {\vphantom {{{f}''_{W_{8} W_{1} } (x)} {{f}''_{W_{8} W_{2} } (x)}}} \right. \kern-\nulldelimiterspace} {{f}''_{W_{8} W_{2} } (x)}$. After simplifications, we have
\[
g_{W_{8} W_{1} \mathunderscore W_{8} W_{2} } (x)=\frac{7\left(
{\begin{array}{l}
 15+364x^{5/2}+30\sqrt x +492x^{3}+364x^{7/2}+ \\
 +150x^{9/2}+90x+90x^{5}+30x^{11/2}+ \\
 +15x^{6}+257x^{4}+257x^{2}+150x^{3/2} \\
 \end{array}} \right)}{3\left( {\begin{array}{l}
 35+828x^{5/2}+70\sqrt x +1084x^{3}+828x^{7/2}+ \\
 +350x^{9/2}+210x+210x^{5}+70x^{11/2}+ \\
 +35x^{6}+589x^{4}+589x^{2}+350x^{3/2} \\
 \end{array}} \right)},
\]
\[
\beta_{W_{8} W_{1} \mathunderscore W_{8} W_{2} } =g_{W_{8} W_{1}
\mathunderscore W_{8} W_{2} } (1)=\frac{42}{41}
\]
and
\[
\frac{42}{41}D_{W_{8} W_{2} }^{27} -D_{W_{8} W_{1} }^{28}
=\frac{1}{41}\left( {W_{8} +41W_{1} -42W_{2} } \right)=\frac{1}{164}V_{6}
:=\frac{1}{164}b{\kern 1pt}f_{V_{6} } \left( {\frac{a}{b}} \right),
\]
where
\begin{equation}
\label{eq18}
f_{V_{6} } (x)=\frac{\left( {x^{2}+6x^{3/2}+22x+6\sqrt x +1} \right)\left(
{\sqrt x -1} \right)^{6}}{x^{3/2}\left( {x+1} \right)}>0,
\quad
\forall x>0,\,\,x\ne 1.
\end{equation}

\item \textbf{For }$\bf{D_{W_{8} W_{2} }^{27} \le \textstyle{{123} \over {119}}D_{W_{8} W_{3} }^{26}} $\textbf{: }Let us consider a function$g_{W_{8} W_{2} \mathunderscore W_{8} W_{3} } (x)={{f}''_{W_{8} W_{2} } (x)} \mathord{\left/ {\vphantom {{{f}''_{W_{8} W_{2} } (x)} {{f}''_{W_{8} W_{3} } (x)}}} \right. \kern-\nulldelimiterspace} {{f}''_{W_{8} W_{3} } (x)}$. After simplifications, we have
\[
g_{W_{8} W_{2} \mathunderscore W_{8} W_{3} } (x)=\frac{9\left(
{\begin{array}{l}
 35+828x^{5/2}+70\sqrt x +1084x^{3}+828x^{7/2}+ \\
 +350x^{9/2}+210x+210x^{5}+70x^{11/2}+ \\
 +35x^{6}+589x^{4}+589x^{2}+350x^{3/2} \\
 \end{array}} \right)}{7\left( {\begin{array}{l}
 45+1028x^{5/2}+90\sqrt x +1284x^{3}+1028x^{7/2}+ \\
 +450x^{9/2}+270x+270x^{5}+90x^{11/2}+ \\
 +45x^{6}+739x^{4}+739x^{2}+450x^{3/2} \\
 \end{array}} \right)}
\]
and
\[
\beta_{W_{8} W_{2} \mathunderscore W_{8} W_{3} } =g_{W_{8} W_{2}
\mathunderscore W_{8} W_{3} } (1)=\frac{123}{119}.
\]
\[
\textstyle{{123} \over {119}}D_{W_{8} W_{3} }^{26} -D_{W_{8} W_{2} }^{27}
=\frac{1}{119}\left( {4W_{8} +119W_{2} -123K_{3} }
\right)=\frac{1}{119}V_{6} .
\]

\item \textbf{For }$\bf{D_{W_{8} W_{3} }^{26} \le \textstyle{{85} \over {81}}D_{W_{8} W_{4} }^{25}} $\textbf{: }Let us consider a function $g_{W_{8} W_{3} \mathunderscore W_{8} W_{4} } (x)={{f}''_{W_{8} W_{3} } (x)} \mathord{\left/ {\vphantom {{{f}''_{W_{8} W_{3} } (x)} {{f}''_{W_{8} W_{4} } (x)}}} \right. \kern-\nulldelimiterspace} {{f}''_{W_{8} W_{4} } (x)}$. After simplifications, we have
\[
g_{W_{8} W_{3} \mathunderscore W_{8} W_{4} } (x)=\frac{5\left(
{\begin{array}{l}
 45+1028x^{5/2}+90\sqrt x +1284x^{3}+1028x^{7/2}+ \\
 +450x^{9/2}+270x+270x^{5}+90x^{11/2}+ \\
 +45x^{6}+739x^{4}+739x^{2}+450x^{3/2} \\
 \end{array}} \right)}{3\left( {\begin{array}{l}
 75+150\sqrt x +1189x^{2}+450x^{5}+1884x^{3}+ \\
 +75x^{6}+150x^{11/2}+450x+1189x^{4}+ \\
 +750x^{3/2}+750x^{9/2}+1628x^{7/2}+1628x^{5/2} \\
 \end{array}} \right)},
\]
\[
\beta_{W_{8} W_{3} \mathunderscore W_{8} W_{4} } =g_{W_{8} W_{3}
\mathunderscore W_{8} W_{4} } (1)=\frac{85}{81}.
\]
and
\[
\frac{85}{81}D_{W_{8} W_{4} }^{25} -D_{W_{8} W_{3} }^{26}
=4W_{8} +81W_{3} -85W_{4} =\frac{1}{81}V_{6} .
\]

\item \textbf{For }$\bf{D_{W_{8} W_{4} }^{25} \le \textstyle{{27} \over {25}}D_{W_{8} W_{5} }^{24}} $\textbf{: }Let us consider a function $g_{W_{8} W_{4} \mathunderscore W_{8} W_{5} } (x)={{f}''_{W_{8} W_{4} } (x)} \mathord{\left/ {\vphantom {{{f}''_{W_{8} W_{4} } (x)} {{f}''_{W_{8} W_{5} } (x)}}} \right. \kern-\nulldelimiterspace} {{f}''_{W_{8} W_{5} } (x)}$. After simplifications, we have
\[
g_{W_{8} W_{4} \mathunderscore W_{8} W_{5} } (x)=\frac{\left(
{\begin{array}{l}
 75+150\sqrt x +1189x^{2}+450x^{5}+1884x^{3}+ \\
 +75x^{6}+150x^{11/2}+450x+1189x^{4}+ \\
 +750x^{3/2}+750x^{9/2}+1628x^{7/2}+1628x^{5/2} \\
 \end{array}} \right)}{75\left( {x+1} \right)^{5}\left( {\sqrt x +1}
\right)^{2}},
\]
\[
\beta_{W_{8} W_{4} \mathunderscore W_{8} W_{5} } =g_{W_{8} W_{4}
\mathunderscore W_{8} W_{5} } (x)=\frac{27}{25}
\]
and
\[
\frac{27}{25}D_{W_{8} W_{5} }^{24} -D_{W_{8} W_{4} }^{25}
=\frac{1}{50}\left( {2W_{8} +25W_{4} -27K_{5} } \right)=\frac{1}{50}V_{6} .
\]

\item \textbf{For }$\bf{D_{W_{8} W_{5} }^{24} \le \textstyle{5 \over 4}D_{W_{8} W_{6} }^{23} }$\textbf{: }Let us consider a function $g_{W_{8} W_{5} \mathunderscore W_{8} W_{6} } (x)={{f}''_{W_{8} W_{5} } (x)} \mathord{\left/ {\vphantom {{{f}''_{W_{8} W_{5} } (x)} {{f}''_{W_{8} W_{6} } (x)}}} \right. \kern-\nulldelimiterspace} {{f}''_{W_{8} W_{6} } (x)}$. After simplifications, we have
\[
g_{W_{8} W_{5} \mathunderscore W_{8} W_{6} } (x)=\frac{5\left( {x+1}
\right)^{2}}{5x^{2}+6x+5},
\]
\[
\beta_{W_{8} W_{5} \mathunderscore W_{8} W_{6} } =g_{W_{8} W_{5}
\mathunderscore W_{8} W_{6} } (1)=\frac{5}{4}
\]
and
\[
\frac{5}{4}D_{W_{8} W_{6} }^{23} -D_{W_{8} W_{5} }^{24}
=\frac{1}{4}\left( {W_{8} +4W_{5} -5W_{6} } \right)=\frac{1}{4}V_{7}
:=\frac{1}{4}b{\kern 1pt}f_{V_{7} } \left( {\frac{a}{b}} \right),
\]
where
\begin{equation}
\label{eq19}
f_{V_{7} } (x)=\frac{\left( {x+6\sqrt x +1} \right)\left( {\sqrt x -1}
\right)^{6}}{x^{3/2}}>0,
\quad
\forall x>0,\,\,x\ne 1.
\end{equation}
Since $f_{W_{8} W_{5} \mathunderscore W_{8} W_{6} } (x)>0$, $\forall
x>0,\,\,x\ne 1$, hence proving the required result.

\item \textbf{For }$\bf{D_{W_{8} W_{6} }^{23} \le 2D_{W_{8} W_{7} }^{22} }$\textbf{: }Let us consider a function $g_{W_{8} W_{6} \mathunderscore W_{8} W_{7} } (x)={{f}''_{W_{8} W_{6} } (x)} \mathord{\left/ {\vphantom {{{f}''_{W_{8} W_{6} } (x)} {{f}''_{W_{8} W_{7} } (x)}}} \right. \kern-\nulldelimiterspace} {{f}''_{W_{8} W_{7} } (x)}$. After simplifications, we have
\[
g_{W_{8} W_{6} \mathunderscore W_{8} W_{7} } (x)=\frac{\left(
{15x^{2}+18x+15} \right)\left( {\sqrt x +1}
\right)^{2}}{15x^{3}+14x^{5/2}+13x^{2}+12x^{3/2}+13x+14\sqrt x +15},
\]
\[
\beta_{W_{8} W_{6} \mathunderscore W_{8} W_{7} } =g_{W_{8} W_{6}
\mathunderscore W_{8} W_{7} } (1)=2.
\]
and
\[
2D_{W_{8} W_{7} }^{22} -D_{W_{8} W_{6} }^{23}
=W_{8} +W_{6} -2W_{7} =\frac{1}{4}V_{8} :=\frac{1}{4}b{\kern 1pt}f_{V_{8} }
\left( {\frac{a}{b}} \right),
\]
where
\begin{equation}
\label{eq20}
f_{V_{8} } (x)=\frac{\left( {\sqrt x +1} \right)^{2}\left( {\sqrt x -1}
\right)^{6}}{x^{3/2}}>0,
\quad
\forall x>0,\,\,x\ne 1.
\end{equation}

\item \textbf{For }$\bf{D_{W_{8} W_{6} }^{23} \le \textstyle{1 \over 2}D_{W_{9} W_{1} }^{36}} $\textbf{: }Let us consider a function $g_{W_{8} W_{6} \mathunderscore W_{9} W_{1} } (x)={{f}''_{W_{8} W_{6} } (x)} \mathord{\left/ {\vphantom {{{f}''_{W_{8} W_{6} } (x)} {{f}''_{W_{9} W_{1} } (x)}}} \right. \kern-\nulldelimiterspace} {{f}''_{W_{9} W_{1} } (x)}$. After simplifications, we have
\[
g_{W_{8} W_{6} \mathunderscore W_{9} W_{1} } (x)=\frac{3\sqrt x \left( {x+1}
\right)^{3}\left( {5x^{2}+6x+5} \right)}{4\left( {x+3} \right)\left( {3x+1}
\right)\left( {x^{4}+2x^{3}+6x^{2}+2x+1} \right)},
\]
\[
\beta_{W_{8} W_{6} \mathunderscore W_{9} W_{1} } =g_{W_{8} W_{6}
\mathunderscore W_{9} W_{1} } (1)=\frac{1}{2}.
\]
and
\[
\frac{1}{2}D_{W_{9} W_{1} }^{36} -D_{W_{8} W_{6} }^{23}
=\frac{1}{2}\left( {W_{9} +2W_{6} -K_{1} -2W_{8} } \right)=\frac{1}{16}V_{9}
:=\frac{1}{16}b{\kern 1pt}f_{V_{9} } \left( {\frac{a}{b}} \right),
\]
where
\begin{equation}
\label{eq21}
f_{V_{9} } (x)=\frac{\left( {\sqrt x +1} \right)^{4}\left( {\sqrt x -1}
\right)^{8}}{16x^{2}\left( {x+1} \right)}>0,
\quad
\forall x>0,\,\,x\ne 1.
\end{equation}

\item \textbf{For }$\bf{D_{W_{9} W_{1} }^{36} \le \textstyle{{56} \over {55}}D_{W_{9} W_{2} }^{35}} $\textbf{: }Let us consider a function $g_{W_{9} W_{1} \mathunderscore W_{9} W_{2} } (x)={{f}''_{W_{9} W_{1} } (x)} \mathord{\left/ {\vphantom {{{f}''_{W_{9} W_{1} } (x)} {{f}''_{W_{9} W_{2} } (x)}}} \right. \kern-\nulldelimiterspace} {{f}''_{W_{9} W_{2} } (x)}$. After simplifications, we have

\[
g_{W_{9} W_{1} \mathunderscore W_{9} W_{2} } (x)=\frac{7\left( {x+3}
\right)\left( {3x+1} \right)\left( {x^{4}+2x^{3}+6x^{2}+2x+1} \right)\left(
{\sqrt x +1} \right)^{2}}{\left( {\begin{array}{l}
 21+42\sqrt x +42x^{13/2}+21x^{7}+399x^{5}+775x^{3}+ \\
 +399x^{2}+775x^{4}+960x^{7/2}+566x^{5/2}+133x+ \\
 +224x^{3/2}+224x^{11/2}+566x^{9/2}+133x^{6} \\
 \end{array}} \right)},
\]
\[
\beta_{W_{9} W_{1} \mathunderscore W_{9} W_{2} } =g_{W_{9} W_{1}
\mathunderscore W_{9} W_{2} } (x)=\frac{56}{55}
\]
and
\[
\frac{56}{55}D_{W_{9} W_{2} }^{35} -D_{W_{9} W_{1} }^{36}
=\frac{1}{55}\left( {W_{9} +55W_{2} -56K_{2} } \right)=\frac{1}{440}V_{10}
=\frac{1}{440}b{\kern 1pt}f_{V_{10} } \left( {\frac{a}{b}} \right),
\]
where
\begin{equation}
\label{eq22}
f_{V_{10} } (x)=\frac{\left( {\begin{array}{l}
 x^{3}+6x^{5/2}+23x^{2}+ \\
 +68x^{3/2}+23x+6\sqrt x +1 \\
 \end{array}} \right)\left( {\sqrt x -1} \right)^{6}}{x^{2}\left( {x+1}
\right)}>0,
\forall x>0,\,\,x\ne 1.
\end{equation}

\item \textbf{For }$\bf{D_{W_{9} W_{2} }^{35} \le \textstyle{{165} \over {161}}D_{W_{9} W_{3} }^{34}} $\textbf{: }Let us consider a function $g_{W_{9} W_{2} \mathunderscore W_{9} W_{3} } (x)={{f}''_{W_{9} W_{2} } (x)} \mathord{\left/ {\vphantom {{{f}''_{W_{9} W_{2} } (x)} {{f}''_{W_{9} W_{3} } (x)}}} \right. \kern-\nulldelimiterspace} {{f}''_{W_{9} W_{3} } (x)}$. After simplifications, we have
\[
g_{W_{9} W_{2} \mathunderscore W_{9} W_{3} } (x)=\frac{3\left(
{\begin{array}{l}
 21+42\sqrt x +42x^{13/2}+21x^{7}+399x^{5}+775x^{3}+ \\
 +399x^{2}+775x^{4}+960x^{7/2}+566x^{5/2}+133x+ \\
 +224x^{3/2}+224x^{11/2}+566x^{9/2}+133x^{6} \\
 \end{array}} \right)}{7\left( {\begin{array}{l}
 9+18\sqrt x +9x^{7}+323x^{4}+171x^{5}+57x^{6}+ \\
 +323x^{3}+171x^{2}+18x^{13/2}+96x^{11/2}+ \\
 +238x^{9/2}+238x^{5/2}+57x+96x^{3/2}+384x^{7/2} \\
 \end{array}} \right)},
\]
\[
\beta_{W_{9} W_{2} \mathunderscore W_{9} W_{3} } =g_{W_{9} W_{2}
\mathunderscore W_{9} W_{3} } (1)=\frac{165}{161}.
\]
and
\[
\frac{165}{161}D_{W_{9} W_{3} }^{34} -D_{W_{9} W_{2} }^{35}
=\frac{1}{161}\left( {4W_{9} +161W_{2} -165W_{3} }
\right)=\frac{1}{322}V_{10} .
\]

\item \textbf{For }$\bf{D_{W_{9} W_{3} }^{34} \le \textstyle{{115} \over {111}}D_{W_{9} W_{4} }^{33} }$\textbf{: }Let us consider a function $g_{W_{9} W_{3} \mathunderscore W_{9} W_{4} } (x)={{f}''_{W_{9} W_{3} } (x)} \mathord{\left/ {\vphantom {{{f}''_{W_{9} W_{3} } (x)} {{f}''_{W_{9} W_{4} } (x)}}} \right. \kern-\nulldelimiterspace} {{f}''_{W_{9} W_{4} } (x)}$. After simplifications, we have
\[
g_{W_{9} W_{3} \mathunderscore W_{9} W_{4} } (x)=\frac{5\left(
{\begin{array}{l}
 9+18\sqrt x +9x^{7}+323x^{4}+171x^{5}+57x^{6}+ \\
 +323x^{3}+171x^{2}+18x^{13/2}+96x^{11/2}+ \\
 +238x^{9/2}+238x^{5/2}+57x+96x^{3/2}+384x^{7/2} \\
 \end{array}} \right)}{3\left( {\begin{array}{l}
 15+30\sqrt x +95x+386x^{9/2}+576x^{7/2}+ \\
 +160x^{3/2}+95x^{6}+30x^{13/2}+160x^{11/2}+15x^{7}+ \\
 +517x^{4}+285x^{5}+517x^{3}+285x^{2}+386x^{5/2} \\
 \end{array}} \right)},
\]
\[
\beta_{W_{9} W_{3} \mathunderscore W_{9} W_{4} } =g_{W_{9} W_{3}
\mathunderscore W_{9} W_{4} } (1)=\frac{115}{111}
\]
and
\[
\frac{115}{111}D_{W_{9} W_{4} }^{33} -D_{W_{9} W_{3} }^{34}
=\frac{1}{111}\left( {4W_{9} +111W_{3} -115W_{4} }
\right)=\frac{1}{222}V_{10} .
\]

\item \textbf{For }$\bf{D_{W_{9} W_{4} }^{33} \le \textstyle{{37} \over {35}}D_{W_{9} W_{5} }^{32} }$\textbf{: }Let us consider a function $g_{W_{9} W_{4} \mathunderscore W_{9} W_{5} } (x)={{f}''_{W_{9} W_{4} } (x)} \mathord{\left/ {\vphantom {{{f}''_{W_{9} W_{4} } (x)} {{f}''_{W_{9} W_{5} } (x)}}} \right. \kern-\nulldelimiterspace} {{f}''_{W_{9} W_{5} } (x)}$. After simplifications, we have
\[
g_{W_{9} W_{4} \mathunderscore W_{9} W_{5} } (x)=\frac{\left(
{\begin{array}{l}
 15+30\sqrt x +95x+386x^{9/2}+576x^{7/2}+ \\
 +160x^{3/2}+95x^{6}+30x^{13/2}+160x^{11/2}+15x^{7}+ \\
 +517x^{4}+285x^{5}+517x^{3}+285x^{2}+386x^{5/2} \\
 \end{array}} \right)}{5\left( {x+1} \right)^{3}\left( {\begin{array}{l}
 3x^{4}+6x^{7/2}+10x^{3}+14x^{5/2}+ \\
 +18x^{2}+14x^{3/2}+10x+6\sqrt x +3 \\
 \end{array}} \right)},
\]
\[
\beta_{W_{9} W_{4} \mathunderscore W_{9} W_{5} } =g_{W_{9} W_{4}
\mathunderscore W_{9} W_{5} } (1)=\frac{37}{35}
\]
and
\[
\frac{37}{35}D_{W_{9} W_{5} }^{32} -D_{W_{9} W_{4} }^{33}
=\frac{1}{35}\left( {2W_{9} +35W_{4} -37W_{5} } \right)=\frac{1}{140}V_{10}
.
\]

\item \textbf{For }$\bf{D_{W_{9} W_{5} }^{32} \le \textstyle{7 \over 6}D_{W_{9} W_{6} }^{31}} $\textbf{: }Let us consider a function $g_{W_{9} W_{5} \mathunderscore W_{9} W_{6} } (x)={{f}''_{W_{9} W_{5} } (x)} \mathord{\left/ {\vphantom {{{f}''_{W_{9} W_{5} } (x)} {{f}''_{W_{9} W_{6} } (x)}}} \right. \kern-\nulldelimiterspace} {{f}''_{W_{9} W_{6} } (x)}$. After simplifications, we have
\[
g_{W_{9} W_{5} \mathunderscore W_{9} W_{6} } (x)=\frac{\left(
{\begin{array}{l}
 3x^{4}+6x^{7/2}+10x^{3}+14x^{5/2}+ \\
 +18x^{2}+14x^{3/2}+10x+6\sqrt x +3 \\
 \end{array}} \right)}{\left( {x+\sqrt x +1} \right)\left(
{\begin{array}{l}
 3x^{3}+3x^{5/2}+4x^{2}+ \\
 +4x^{3/2}+4x+3\sqrt x +3 \\
 \end{array}} \right)},
\]
\[
\beta_{W_{9} W_{5} \mathunderscore W_{9} W_{6} } =g_{W_{9} W_{5}
\mathunderscore W_{9} W_{6} } (1)=\frac{7}{6}
\]
and
\[
\frac{7}{6}D_{W_{9} W_{6} }^{31} -D_{W_{9} W_{5} }^{32}
=\frac{1}{6}\left( {W_{9} +6W_{5} -7W_{6} } \right)=\frac{1}{48}V_{11}
:=\frac{1}{48}b{\kern 1pt}f_{V_{11} } \left( {\frac{a}{b}} \right),
\]
where
\begin{equation}
\label{eq23}
f_{V_{11} } (x)=\frac{\left( {x^{2}+6x^{3/2}+22x+6\sqrt x +1} \right)\left(
{\sqrt x -1} \right)^{6}}{x^{2}}>0,
\quad
\forall x>0,\,\,x\ne 1.
\end{equation}

\item \textbf{For }$\bf{D_{W_{8} W_{7} }^{22} \le \textstyle{1 \over 3}D_{W_{9} W_{6} }^{31}} $\textbf{: }Let us consider a function $g_{W_{8} W_{2} \mathunderscore W_{9} W_{6} } (x)={{f}''_{W_{8} W_{2} } (x)} \mathord{\left/ {\vphantom {{{f}''_{W_{8} W_{2} } (x)} {{f}''_{W_{9} W_{6} } (x)}}} \right. \kern-\nulldelimiterspace} {{f}''_{W_{9} W_{6} } (x)}$. After simplifications, we have
\[
g_{W_{8} W_{2} \mathunderscore W_{9} W_{6} } (x)=\frac{\sqrt x \left(
{\begin{array}{l}
 15x^{3}+14x^{5/2}+13x^{2}+ \\
 +12x^{3/2}+13x+14\sqrt x +15 \\
 \end{array}} \right)}{4\left( {x+\sqrt x +1} \right)\left(
{\begin{array}{l}
 3x^{3}+3x^{5/2}+4x^{2}+ \\
 +4x^{3/2}+4x+3\sqrt x +3 \\
 \end{array}} \right)},
\]
\[
\beta_{W_{8} W_{2} \mathunderscore W_{9} W_{6} } =g_{W_{8} W_{2}
\mathunderscore W_{9} W_{6} } (1)=\frac{1}{3}
\]
and
\[
\frac{1}{3}D_{W_{9} W_{6} }^{31} -D_{W_{8} W_{7} }^{22}
=\frac{1}{6}\left( {W_{9} +3W_{7} -W_{6} -3W_{8} }
\right)=\frac{1}{24}V_{12} :=\frac{1}{24}b{\kern 1pt}f_{V_{3} } \left(
{\frac{a}{b}} \right),
\]
where
\begin{equation}
\label{eq24}
f_{V_{12} } (x)=\frac{\left( {\sqrt x +1} \right)^{2}\left( {\sqrt x -1}
\right)^{8}}{x^{2}}>0,
\quad
\forall x>0,\,\,x\ne 1.
\end{equation}

\item \textbf{For }$\bf{D_{W_{9} W_{6} }^{31} \le \textstyle{3 \over 2}D_{W_{9} W_{7} }^{30} }$\textbf{: }Let us consider a function $g_{W_{9} W_{6} \mathunderscore W_{9} W_{7} } (x)={{f}''_{W_{9} W_{6} } (x)} \mathord{\left/ {\vphantom {{{f}''_{W_{9} W_{6} } (x)} {{f}''_{W_{9} W_{7} } (x)}}} \right. \kern-\nulldelimiterspace} {{f}''_{W_{9} W_{7} } (x)}$. After simplifications, we have
\[
g_{W_{9} W_{6} \mathunderscore W_{9} W_{7} } (x)=\frac{\left( {x+\sqrt x +1}
\right)\left( {\begin{array}{l}
 3x^{3}+3x^{5/2}+4x^{2}+ \\
 +4x^{3/2}+4x+3\sqrt x +3 \\
 \end{array}} \right)}{3\left( {x+1} \right)\left( {x^{2}+1} \right)\left(
{\sqrt x +1} \right)^{2}},
\]
\[
\beta_{W_{9} W_{6} \mathunderscore W_{9} W_{7} } =g_{W_{9} W_{6}
\mathunderscore W_{9} W_{7} } (1)=\frac{3}{2}
\]
and
\[
\frac{3}{2}D_{W_{9} W_{7} }^{30} -D_{W_{9} W_{6} }^{31}
=\frac{1}{2}\left( {W_{9} +2W_{6} -3W_{7} } \right)=\frac{1}{16}V_{13}
:=b{\kern 1pt}f_{V_{13} } \left( {\frac{a}{b}} \right),
\]
with
\begin{equation}
\label{eq25}
f_{V_{13} } (x)=\frac{\left( {x-1} \right)^{2}\left( {x+4\sqrt x +1}
\right)\left( {\sqrt x -1} \right)^{4}}{x^{2}}>0,
\quad
\forall x>0,\,\,x\ne 1.
\end{equation}

\item \textbf{For }$\bf{D_{W_{9} W_{7} }^{30} \le 2D_{W_{9} W_{8} }^{29}} $\textbf{: }Let us consider a function $g_{W_{9} W_{7} \mathunderscore W_{9} W_{8} } (x)={{f}''_{W_{9} W_{7} } (x)} \mathord{\left/ {\vphantom {{{f}''_{W_{9} W_{7} } (x)} {{f}''_{W_{9} W_{8} } (x)}}} \right. \kern-\nulldelimiterspace} {{f}''_{W_{9} W_{8} } (x)}$. After simplifications, we have
\[
g_{W_{9} W_{7} \mathunderscore W_{9} W_{8} } (x)=\frac{12\left( {x+1}
\right)\left( {x^{2}+1} \right)\left( {\sqrt x +1} \right)^{2}}{\left(
{\begin{array}{l}
 12x^{4}+9x^{7/2}+10x^{3}+11x^{5/2}+ \\
 +12x^{2}+11x^{3/2}+10x+9\sqrt x +12 \\
 \end{array}} \right)},
\]
\[
\beta_{W_{9} W_{7} \mathunderscore W_{9} W_{8} } =g_{W_{9} W_{7}
\mathunderscore W_{9} W_{8} } (1)=2
\]
and
\[
2D_{W_{9} W_{8} }^{29} -D_{W_{9} W_{7} }^{30}
=W_{9} +W_{7} -2W_{8} =\frac{1}{8}V_{14} :=\frac{1}{8}b{\kern 1pt}f_{V_{14}
} \left( {\frac{a}{b}} \right),
\]
where
\begin{equation}
\label{eq26}
f_{V_{14} } (x)=\frac{\left( {x+1} \right)\left( {\sqrt x +1}
\right)^{2}\left( {\sqrt x -1} \right)^{6}}{8x^{2}}>0,
\quad
\forall x>0,\,\,x\ne 1.
\end{equation}
\end{enumerate}

Combining the results 1-27, we get the proof of (\ref{eq12}).
\end{proof}

\begin{remark} Based on the equalities given in (\ref{eq11a}) we have the following proportionality relations among the six means appearing in Section 1:
\begin{multicols}{2}
\begin{enumerate}
\item $4A=2(C+H)=3R+H;$
\item $3R=C+2A=2C+H;$
\item $3N=2A+G;$
\item $3C+2H=3R+2A;$
\item $C+6A=H+6R;$
\item $C+3N=G+3R;$
\item $3N+2A=2C+2H+G;$
\item $27R+2G=14A+9C+6N;$
\item $3\left( {N+3R} \right)=8A+3C+G;$
\item $3G+8H+9C=3R+8A+9N;$
\item $4G+14H+17C=9R+14A+12N;$
\item $5G+24H+31C=21R+24A+15N.$
\end{enumerate}
\end{multicols}
\end{remark}

\subsection{Reverse Inequalities}

We observe from the above results that the first four inequalities appearing
in pyramid are equal with some multiplicative constants. The other four
inequalities satisfies reverse inequalities given by
\begin{align}
& 1. \quad D_{W_{6} W_{5} }^{11} \le D_{W_{6} W_{4} }^{12} \le D_{W_{6} W_{3} }^{13} \le D_{W_{6} W_{2} }^{14} \le D_{W_{6} W_{1} }^{15} \le \notag\\
& \hspace{25pt} \le \textstyle{{14} \over {13}}D_{W_{6} W_{2} }^{14} \le \textstyle{6 \over
5}D_{W_{6} W_{3} }^{13} \le \textstyle{{10} \over 7}D_{W_{6} W_{4} }^{12}
\le 2D_{W_{6} W_{5} }^{11};\notag\\\notag\\
& 2. \quad D_{W_{7} W_{6} }^{16} \le D_{W_{7} W_{5} }^{17} \le D_{W_{7} W_{4} }^{18} \le D_{W_{7} W_{3} }^{19} \le D_{W_{7} W_{2} }^{20} \le D_{W_{7} W_{1} }^{21} \le \notag\\
& \hspace{25pt} \le \textstyle{{28} \over {27}}D_{W_{7} W_{2} }^{20} \le \textstyle{{12}
\over {11}}D_{W_{7} W_{3} }^{19} \le \textstyle{{20} \over {17}}D_{W_{7}
W_{4} }^{18} \le \textstyle{4 \over 3}D_{W_{7} W_{5} }^{17} \le 2D_{W_{7}
W_{6} }^{16} ;\notag\\\notag\\
& 3. \quad D_{W_{8} W_{7} }^{22} \le D_{W_{8} W_{6} }^{23} \le D_{W_{8} W_{5} }^{24} \le D_{W_{8} W_{4} }^{25} \le D_{W_{8} W_{3} }^{26} \le D_{W_{8} W_{2} }^{27} \le D_{W_{8} W_{1} }^{28} \le \notag\\
& \hspace{25pt} \le \textstyle{{42} \over {41}}D_{W_{8} W_{2} }^{27} \le \textstyle{{18}
\over {17}}D_{W_{8} W_{3} }^{26} \le \textstyle{{10} \over 9}D_{W_{8} W_{4}
}^{25} \le \textstyle{6 \over 5}D_{W_{8} W_{5} }^{24} \le \textstyle{3 \over
2}D_{W_{8} W_{6} }^{23} \le 3D_{W_{8} W_{7} }^{22};\notag\\\notag\\
& 4.\quad D_{W_{9} W_{8} }^{29} \le D_{W_{9} W_{7} }^{30} \le D_{W_{9} W_{6} }^{31} \le D_{W_{9} W_{5} }^{32} \le D_{W_{9} W_{4} }^{33} \le D_{W_{9} W_{3} }^{34} \le \notag\\
& \hspace{25pt} \le D_{W_{9} W_{2} }^{35} \le D_{W_{9} W_{1} }^{36} \le \textstyle{{56} \over {55}}D_{W_{9} W_{2} }^{35} \le \textstyle{{24} \over {23}}D_{W_{9} W_{3} }^{34} \le \textstyle{{40} \over {37}}D_{W_{9}
W_{4} }^{33} \le \notag\\
& \hspace{35pt} \le \textstyle{8 \over 7}D_{W_{9} W_{5} }^{32} \le \textstyle{4
\over 3}D_{W_{9} W_{6} }^{31} \le 2D_{W_{9} W_{7} }^{30} \le 4D_{W_{9} W_{8}
}^{29} .\notag
\end{align}

\begin{remark} It is interesting to observe that in first and second case there is difference of only 2 in between first and last elements. While, in the third case is of 3 and finally, in the forth case is of 4.
\end{remark}

\subsection{Second Stage}

In this stage we shall bring inequalities based on measures arising due to
first stage. The above 27 parts generate some new measures given by
\begin{equation}
\label{eq27}
V_{t} (P\vert \vert Q):=\sum\limits_{i=1}^n {q_{i} f_{V_{t} } } \left(
{\frac{p_{i} }{q_{i} }} \right),
\quad
t=1,2,...,14,
\end{equation}
where $f_{V_{t} } (x)$, $t=1,2,...,14$ are as given by (\ref{eq13})-(\ref{eq26})
respectively. In all the cases we have $f_{V_{t} } (1)=0$, $t=1,2,...,14$.
By the application of Lemma 1.1, we can say that the above 14 measures are
convex. We shall try to connect 14 measures given in (\ref{eq27}) through
inequalities.

\begin{theorem} The following inequalities hold:
\begin{equation}
\label{eq28}
V_{1} \le \textstyle{1 \over 8}V_{3} \le \left\{ {\begin{array}{l}
 \textstyle{1 \over 2}V_{4} \le \textstyle{1 \over {16}}V_{7} \le
\textstyle{1 \over 8}V_{8} \\
 \textstyle{1 \over {36}}V_{6} \le \textstyle{1 \over {128}}V_{10} \\
 \end{array}} \right\}\le \textstyle{1 \over {72}}V_{11} \le \textstyle{1
\over {48}}V_{13} \le \textstyle{1 \over {16}}V_{14}
\end{equation}
and
\begin{equation}
\label{eq29}
V_{2} \le \textstyle{1 \over 4}V_{5} \le \textstyle{1 \over {16}}V_{9} \le
\textstyle{1 \over 8}V_{12} .
\end{equation}
\end{theorem}

\begin{proof} We shall prove the above theorem following the similar lines
of Theorem 2.1. Since, we need the second derivatives of the functions given
by (\ref{eq13})-(\ref{eq26}) to prove the theorem, here below are their values:

\begin{align}
& {f}''_{V_{1} } (x)=\frac{\left( {\sqrt x -1} \right)^{4}\left(
{\begin{array}{l}
 3x^{3}+12x^{5/2}+25x^{2}+ \\
 +40x^{3/2}+25x+12\sqrt x +3 \\
 \end{array}} \right)}{4x^{5/2}\left( {x+1} \right)^{3}},\notag\\
& {f}''_{V_{2} } (x)=\frac{2\left( {\sqrt x -1} \right)^{6}\left(
{\begin{array}{l}
 x^{3}+3x^{5/2}+6x^{2}+ \\
 +8x^{3/2}+6x+3\sqrt x +1 \\
 \end{array}} \right)}{x^{3}\left( {x+1} \right)^{3}},\notag\\
& {f}''_{V_{3} } (x)=\frac{2\left( {\sqrt x -1} \right)^{4}\left(
{\begin{array}{l}
 x^{4}+4x^{7/2}+13x^{3}+24x^{5/2}+ \\
 +36x^{2}+24x^{3/2}+13x+4\sqrt x +1 \\
 \end{array}} \right)}{x^{3}\left( {x+1} \right)^{3}},]\notag\\
& {f}''_{V_{4} } (x)=\frac{\left( {\sqrt x -1} \right)^{4}\left( {4x+7\sqrt x
+4} \right)}{4x^{3}},\notag\\
& {f}''_{V_{5} } (x)=\frac{\left( {\sqrt x -1} \right)^{6}\left(
{\begin{array}{l}
 15x^{4}+42x^{7/2}+108x^{3}+174x^{5/2}+ \\
 +218x^{2}+174x^{3/2}+108x+42\sqrt x +15 \\
 \end{array}} \right)}{4x^{7/2}\left( {x+1} \right)^{3}},\notag\\
 & {f}''_{V_{6} } (x)=\frac{5\left( {\sqrt x -1} \right)^{4}\left(
{\begin{array}{l}
 3x^{5}+12x^{9/2}+39x^{4}+96x^{7/2}+ \\
 +166x^{3}+232x^{5/2}+166x^{2}+ \\
 +96x^{3/2}+39x+12\sqrt x +3 \\
 \end{array}} \right)}{4x^{7/2}\left( {x+1} \right)^{3}},\notag\\
& {f}''_{V_{7} } (x)=\frac{15\left( {\sqrt x -1} \right)^{4}\left(
{x^{2}+4x^{3/2}+x^{2}+6x+4\sqrt x +1} \right)}{4x^{7/2}},\notag\\
& {f}''_{V_{8} } (x)=\frac{\left( {\sqrt x -1} \right)^{4}\left(
{15x^{2}+28x^{3/2}+34x+28\sqrt x +15} \right)}{4x^{7/2}},\notag\\
& {f}''_{V_{9} } (x)=\frac{\left( {\sqrt x +1} \right)^{2}\left( {\sqrt x -1}
\right)^{6}\left( {\begin{array}{l}
 6x^{4}+9x^{7/2}+32x^{3}+35x^{5/2}+ \\
 +60x^{2}+35x^{3/2}+32x^{2}+9\sqrt x +6 \\
 \end{array}} \right)}{x^{4}\left( {x+1} \right)^{3}}, \notag\\
& {f}''_{V_{10} } (x)=\frac{2\left( {\sqrt x -1} \right)^{4}\left(
{\begin{array}{l}
 3x^{6}+12x^{11/2}+40x^{5}+100x^{9/2}+217x^{4}+ \\
 +352x^{7/2}+472x^{3}+352x^{5/2}+217x^{2}+ \\
 +100x^{3/2}+40x+12\sqrt x +3 \\
 \end{array}} \right)}{x^{4}\left( {x+1} \right)^{3}},\notag\\
& {f}''_{V_{11} } (x)=\frac{2\left( {\sqrt x -1} \right)^{4}\left(
{\begin{array}{l}
 3x^{3}+12x^{5/2}+31x^{2}+ \\
 +43x^{3/2}+31x+12\sqrt x +3 \\
 \end{array}} \right)}{x^{4}}\notag\\
& {f}''_{V_{12} } (x)=\frac{\left( {\sqrt x -1} \right)^{6}\left(
{12x^{2}+27x^{3/2}+34x+27\sqrt x +12} \right)}{2x^{4}},\notag\\
& {f}''_{V_{13} } (x)=\frac{2\left( {\sqrt x -1} \right)^{4}\left(
{\begin{array}{l}
 3x^{3}+12x^{5/2}+19x^{2}+ \\
 +22x^{3/2}+19x+12\sqrt x +3 \\
 \end{array}} \right)}{x^{4}}\notag
\intertext{and}
& {f}''_{V_{14} } (x)=\frac{\left( {\sqrt x -1} \right)^{4}\left(
{\begin{array}{l}
 6x^{3}+9x^{5/2}+10x^{2}+ \\
 +10x^{3/2}+10x+9\sqrt x +6 \\
 \end{array}} \right)}{x^{4}}.\notag
\end{align}

We shall prove the above theorem by parts. In view of procedure used in
Theorem 2.1, we shall write the proof of each part in very summarized way.

\begin{enumerate}
\item \textbf{For }$\bf{V_{1} (a,b)\le \textstyle{1 \over 8}V_{3} (a,b)} $\textbf{: }Let us consider a function $g_{V_{1} \mathunderscore V_{3} } (x)={{f}''_{V_{1} } (x)} \mathord{\left/ {\vphantom {{{f}''_{V_{1} } (x)} {{f}''_{V_{3} } (x)}}} \right. \kern-\nulldelimiterspace} {{f}''_{V_{3} } (x)}$. After simplifications, we have
\[
g_{V_{1} \mathunderscore V_{3} } (x)=\frac{x\left( {\begin{array}{l}
 3x^{3}+12x^{5/2}+25x^{2}+ \\
 +40x^{3/2}+25x+12\sqrt x +3 \\
 \end{array}} \right)}{8\left( {\begin{array}{l}
 24x^{3}+\sqrt x +36x^{5/2}+13x^{3/2}+ \\
 +24x^{2}+4x+13x^{7/2}+x^{9/2}+4x^{4} \\
 \end{array}} \right)},
\]
\[
\beta_{V_{1} \mathunderscore V_{3} } =g_{V_{1} \mathunderscore V_{3} }
(1)=\frac{1}{8}
\]
and
\[
\textstyle{1 \over 8}V_{3} (a,b)-V_{1} (a,b)=\frac{1}{8}U_{1}
(a,b)=\frac{1}{8}b{\kern 1pt}f_{U_{1} } \left( {\frac{a}{b}} \right),
\]
where
\begin{equation}
\label{eq30}
f_{U_{1} } (x)=\frac{\left( {\sqrt x -1} \right)^{8}}{x\left( {x+1}
\right)}>0,
\quad
\forall x>0,\,\,x\ne 1.
\end{equation}

\item \textbf{For }$\bf{V_{3} (a,b)\le 4V_{4} (a,b)}$\textbf{: }Let us consider a function $g_{V_{3} \mathunderscore V_{4} } (x)={{f}''_{V_{3} } (x)} \mathord{\left/ {\vphantom {{{f}''_{V_{3} } (x)} {{f}''_{V_{4} } (x)}}} \right. \kern-\nulldelimiterspace} {{f}''_{V_{4} } (x)}$. After simplifications, we have
\[
g_{V_{3} \mathunderscore V_{4} } (x)=\frac{4\left( {\begin{array}{l}
 24x^{3}+\sqrt x +36x^{5/2}+13x^{3/2}+ \\
 +24x^{2}+4x+13x^{7/2}+x^{9/2}+4x^{4} \\
 \end{array}} \right)}{\left( {4x^{3/2}+7x+4\sqrt x } \right)\left( {x+1}
\right)^{3}},
\]
\[
\beta_{V_{3} \mathunderscore V_{4} } =g_{V_{3} \mathunderscore V_{4} }
(1)=4
\]
and
\[
4V_{4} (a,b)-V_{3} (a,b)=3U_{1} (a,b).
\]

\item \textbf{For }$\bf{V_{4} (a,b)\le \textstyle{1 \over 8}V_{7} (a,b)}$\textbf{: }Let us consider a function $g_{V_{4} \mathunderscore V_{7} } (x)={{f}''_{V_{4} } (x)} \mathord{\left/ {\vphantom {{{f}''_{V_{4} } (x)} {{f}''_{V_{7} } (x)}}} \right. \kern-\nulldelimiterspace} {{f}''_{V_{7} } (x)}$. After simplifications, we have
\[
g_{V_{4} \mathunderscore V_{7} } (x)=\frac{2\left( {4x^{3/2}+7x+4\sqrt x }
\right)}{15\left( {4\sqrt x +6x+1+x^{2}+4x^{3/2}} \right)},
\]
\[
\beta_{V_{4} \mathunderscore V_{7} } =g_{V_{4} \mathunderscore V_{7} }
(1)=\frac{1}{8}
\]
and
\[
\frac{1}{8}V_{7} (a,b)-V_{4} (a,b)=\frac{1}{8}U_{2} (a,b)=\frac{1}{8}b{\kern
1pt}f_{U_{2} } \left( {\frac{a}{b}} \right),
\]
where
\begin{equation}
\label{eq31}
f_{U_{2} } (x)=\frac{\left( {\sqrt x -1} \right)^{8}}{x^{3/2}}>0,
\quad
\forall x>0,\,\,x\ne 1.
\end{equation}

\item \textbf{For }$\bf{V_{3} (a,b)\le \textstyle{2 \over 9}V_{6} (a,b)} $\textbf{: }Let us consider a function $g_{V_{3} \mathunderscore V_{6} } (x)={{f}''_{V_{3} } (x)} \mathord{\left/ {\vphantom {{{f}''_{V_{3} } (x)} {{f}''_{V_{6} } (x)}}} \right. \kern-\nulldelimiterspace} {{f}''_{V_{6} } (x)}$. After simplifications, we have
\[
g_{V_{3} \mathunderscore V_{6} } (x)=\frac{8\left( {\begin{array}{l}
 24x^{3}+\sqrt x +36x^{5/2}+13x^{3/2}+ \\
 +24x^{2}+4x+13x^{7/2}+x^{9/2}+4x^{4} \\
 \end{array}} \right)}{5\left( {\begin{array}{l}
 166x^{2}+96x^{3/2}+12\sqrt x +39x+3+3x^{5}+ \\
 +232x^{5/2}+12x^{9/2}+96x^{7/2}+166x^{3}+39x^{4} \\
 \end{array}} \right)},
\]
\[
\beta_{V_{3} \mathunderscore V_{6} } =g_{V_{3} \mathunderscore V_{6} }
(1)=\frac{2}{9}
\]
and
\[
\frac{2}{9}V_{6} (a,b)-V_{3} (a,b)=\frac{1}{9}U_{3} (a,b)=\frac{1}{9}b{\kern
1pt}f_{U_{3} } \left( {\frac{a}{b}} \right),
\]
where
\begin{equation}
\label{eq32}
f_{U_{3} } (x)=\frac{\left( {\sqrt x -1} \right)^{8}\left( {2x+7\sqrt x +2}
\right)}{x^{3/2}\left( {x+1} \right)}>0,
\quad
\forall x>0,\,\,x\ne 1.
\end{equation}

\item \textbf{For }$\bf{V_{6} (a,b)\le \textstyle{9 \over {32}}V_{10} (a,b)}$\textbf{: }Let us consider a function $g_{V_{6} \mathunderscore V_{10} } (x)={{f}''_{V_{6} } (x)} \mathord{\left/ {\vphantom {{{f}''_{V_{6} } (x)} {{f}''_{V_{10} } (x)}}} \right. \kern-\nulldelimiterspace} {{f}''_{V_{10} } (x)}$. After simplifications, we have
\[
g_{V_{6} \mathunderscore V_{10} } (x)=\frac{5x\left( {\begin{array}{l}
 166x^{2}+96x^{3/2}+12\sqrt x +39x+3+3x^{5}+ \\
 +232x^{5/2}+12x^{9/2}+96x^{7/2}+166x^{3}+39x^{4} \\
 \end{array}} \right)}{8\left( {\begin{array}{l}
 12x^{6}+100x^{2}+40x^{3/2}+3\sqrt x +12x+ \\
 +100x^{5}+217x^{5/2}+217x^{9/2}+472x^{7/2}+ \\
 +352x^{3}+352x^{4}+40x^{11/2}+3x^{13/2} \\
 \end{array}} \right)},
\]
This gives $\beta_{V_{6} \mathunderscore V_{10} } =g_{V_{6} \mathunderscore
V_{10} } (1)=\frac{9}{32}$. Let us consider now,

and
\[
\frac{9}{32}V_{10} (a,b)-V_{6} (a,b)=\frac{1}{32}U_{4}
(a,b)=\frac{1}{32}b{\kern 1pt}f_{U_{4} } \left( {\frac{a}{b}} \right),
\]
where
\begin{equation}
\label{eq33}
f_{U_{4} } (x)=\frac{\left( {\sqrt x -1} \right)^{8}\left(
{9x^{2}+40x^{3/2}+86x+40\sqrt x +9} \right)}{x^{2}\left( {x+1} \right)}>0,
\quad
\forall x>0,\,\,x\ne 1.
\end{equation}

\item \textbf{For }$\bf{V_{10} (a,b)\le \textstyle{{16} \over 9}V_{11} (a,b)}$\textbf{: }Let us consider a function $g_{V_{10} \mathunderscore V_{11} } (x)={{f}''_{V_{10} } (x)} \mathord{\left/ {\vphantom {{{f}''_{V_{10} } (x)} {{f}''_{V_{11} } (x)}}} \right. \kern-\nulldelimiterspace} {{f}''_{V_{11} } (x)}$. After simplifications, we have
\[
g_{V_{10} \mathunderscore V_{11} } (x)=\frac{\left( {\begin{array}{l}
 12x^{6}+100x^{2}+40x^{3/2}+3\sqrt x +12x+ \\
 +100x^{5}+217x^{5/2}+217x^{9/2}+472x^{7/2}+ \\
 +352x^{3}+352x^{4}+40x^{11/2}+3x^{13/2} \\
 \end{array}} \right)}{\left( {\begin{array}{l}
 43x^{2}+31x^{3/2}+3\sqrt x + \\
 +12x+31x^{5/2}+3x^{7/2}+12x^{3} \\
 \end{array}} \right)\left( {x+1} \right)^{3}},
\]
\[
\beta_{V_{10} \mathunderscore V_{11} } =g_{V_{10} \mathunderscore V_{11} }
(1)=\frac{16}{9}
\]
and
\[
\frac{16}{9}V_{11} (a,b)-V_{10} (a,b)=\frac{7}{9}U_{5}
(a,b)=\frac{7}{9}b{\kern 1pt}f_{U_{5} } \left( {\frac{a}{b}} \right),
\]
where
\begin{equation}
\label{eq34}
f_{U_{5} } (x)=\frac{\left( {\sqrt x -1} \right)^{8}\left(
{x^{2}+8x^{3/2}+38x+8\sqrt x +1} \right)}{x^{2}\left( {x+1} \right)}>0,
\quad
\forall x>0,\,\,x\ne 1.
\end{equation}
\item \textbf{For }$\bf{V_{6} (a,b)\le \textstyle{9 \over 4}V_{7} (a,b)}$\textbf{: }Let us consider a function $g_{V_{6} \mathunderscore V_{7} } (x)={{f}''_{V_{6} } (x)} \mathord{\left/ {\vphantom {{{f}''_{V_{6} } (x)} {{f}''_{V_{7} } (x)}}} \right. \kern-\nulldelimiterspace} {{f}''_{V_{7} } (x)}$. After simplifications, we have
\[
g_{V_{6} \mathunderscore V_{7} } (x)=\frac{\left( {\begin{array}{l}
 166x^{2}+96x^{3/2}+12\sqrt x +39x+3+3x^{5}+ \\
 +232x^{5/2}+12x^{9/2}+96x^{7/2}+166x^{3}+39x^{4} \\
 \end{array}} \right)}{3\left( {4\sqrt x +6x+1+x^{2}+4x^{3/2}} \right)\left(
{x+1} \right)^{3}},
\]
\[
\beta_{V_{6} \mathunderscore V_{7} } =g_{V_{6} \mathunderscore V_{7} }
(1)=\frac{9}{4}
\]
and
\[
\frac{9}{4}V_{7} (a,b)-V_{6} (a,b)=\frac{5}{4}U_{6} (a,b)=\frac{5}{4}b{\kern
1pt}f_{U_{6} } \left( {\frac{a}{b}} \right),
\]
where
\begin{equation}
\label{eq35}
f_{U_{6} } (x)=\frac{\left( {\sqrt x -1} \right)^{8}\left( {x+8\sqrt x +1}
\right)}{x^{3/2}\left( {x+1} \right)}>0,
\quad
\forall x>0,\,\,x\ne 1.
\end{equation}

\item \textbf{For }$\bf{V_{7} (a,b)\le 2\,V_{8} (a,b)}$\textbf{: }Let us consider a function $g_{V_{7} \mathunderscore V_{8} } (x)={{f}''_{V_{7} } (x)} \mathord{\left/ {\vphantom {{{f}''_{V_{7} } (x)} {{f}''_{V_{8} } (x)}}} \right. \kern-\nulldelimiterspace} {{f}''_{V_{8} } (x)}$. After simplifications, we have
\[
g_{V_{7} \mathunderscore V_{8} } (x)=\frac{15\left( {4\sqrt x
+6x+1+x^{2}+4x^{3/2}} \right)}{15x^{2}+28x^{3/2}+34x+28\sqrt x +15},
\]
\[
\beta_{V_{7} \mathunderscore V_{8} } =g_{V_{7} \mathunderscore V_{8} }
(1)=2
\]
and
\[
2V_{8} (a,b)-V_{7} =U_{2} (a,b).
\]

\item \textbf{For }$\bf{V_{8} (a,b)\le \textstyle{1 \over 9}\,V_{11} (a,b)}$\textbf{: }Let us consider a function $g_{V_{8} \mathunderscore V_{11} } (x)={{f}''_{V_{8} } (x)} \mathord{\left/ {\vphantom {{{f}''_{V_{8} } (x)} {{f}''_{V_{11} } (x)}}} \right. \kern-\nulldelimiterspace} {{f}''_{V_{11} } (x)}$. After simplifications, we have
\[
g_{V_{8} \mathunderscore V_{11} } (x)=\frac{x\left(
{15x^{2}+28x^{3/2}+34x+28\sqrt x +15} \right)}{8\left( {\begin{array}{l}
 43x^{2}+31x^{3/2}+3\sqrt x + \\
 +12x+31x^{5/2}+3x^{7/2}+12x^{3} \\
 \end{array}} \right)},
\]
\[
\beta_{V_{8} \mathunderscore V_{11} } =g_{V_{8} \mathunderscore V_{11} }
(1)=\frac{1}{9}
\]
and
\[
4V_{2} (a,b)-V_{4} (a,b)=\frac{1}{9}U_{7} (a,b)=\frac{1}{9}b{\kern
1pt}f_{U_{7} } \left( {\frac{a}{b}} \right),
\]
where
\begin{equation}
\label{eq36}
f_{U_{7} } (x)=\frac{\left( {\sqrt x -1} \right)^{8}\left( {x-\sqrt x +1}
\right)}{x^{2}}>0,
\quad
\forall x>0,\,\,x\ne 1.
\end{equation}

\item \textbf{For }$\bf{V_{11} (a,b) \le \textstyle{3 \over 2}\,V_{13} (a,b) }$\textbf{: }Let us consider a function $g_{V_{11} \mathunderscore V_{13} } (x)={{f}''_{V_{11} } (x)} \mathord{\left/ {\vphantom {{{f}''_{V_{11} } (x)} {{f}''_{V_{13} } (x)}}} \right. \kern-\nulldelimiterspace} {{f}''_{V_{13} } (x)}$. After simplifications, we have
\[
g_{V_{11} \mathunderscore V_{13} } (x)=\frac{\left( {\begin{array}{l}
 43x^{2}+31x^{3/2}+3\sqrt x +12x+ \\
 +31x^{5/2}+3x^{7/2}+12x^{3} \\
 \end{array}} \right)}{\sqrt x \left( {\begin{array}{l}
 3x^{3}+12x^{5/2}+19x^{2}+ \\
 +22x^{3/2}+19x+12\sqrt x +3 \\
 \end{array}} \right)},
\]
\[
\beta_{V_{11} \mathunderscore V_{13} } =g_{V_{11} \mathunderscore V_{13} }
(1)=\frac{3}{2}
\]
and
\[
\frac{3}{2}V_{13} (a,b)-V_{11} (a,b)=\frac{1}{2}U_{8}
(a,b)=\frac{1}{2}b{\kern 1pt}f_{U_{8} } \left( {\frac{a}{b}} \right),
\]
where
\begin{equation}
\label{eq37}
f_{U_{8} } (x)=\frac{\left( {\sqrt x -1} \right)^{8}\left( {x+8\sqrt x +1}
\right)}{x^{2}}>0,
\quad
\forall x>0,\,\,x\ne 1.
\end{equation}

\item \textbf{For }$\bf{V_{13} (a,b)\le 3\,V_{14} (a,b)}$\textbf{: }Let us consider a function $g_{V_{13} \mathunderscore V_{14} } (x)={{f}''_{V_{13} } (x)} \mathord{\left/ {\vphantom {{{f}''_{V_{13} } (x)} {{f}''_{V_{14} } (x)}}} \right. \kern-\nulldelimiterspace} {{f}''_{V_{14} } (x)}$. After simplifications, we have
\[
g_{V_{13} \mathunderscore V_{14} } (x)=\frac{2\left( {\begin{array}{l}
 3x^{3}+12x^{5/2}+19x^{2}+ \\
 +22x^{3/2}+19x+12\sqrt x +3 \\
 \end{array}} \right)}{\left( {\begin{array}{l}
 10x^{3/2}+10x+9\sqrt x +6+ \\
 +10x^{2}+9x^{5/2}+6x^{3} \\
 \end{array}} \right)},
\]
\[
\beta_{V_{13} \mathunderscore V_{14} } =g_{V_{13} \mathunderscore V_{14} }
(1)=3
\]
and
\[
3V_{14} (a,b)-V_{13} (a,b)=2U_{9} (a,b)=2b{\kern 1pt}f_{U_{9} } \left(
{\frac{a}{b}} \right),
\]
where
\begin{equation}
\label{eq38}
f_{U_{9} } (x)=\frac{\left( {\sqrt x -1} \right)^{8}\left( {\sqrt x +1}
\right)^{2}}{x^{2}}>0,
\quad
\forall x>0,\,\,x\ne 1.
\end{equation}

\item \textbf{For }$\bf{V_{2} (a,b)\le \textstyle{1 \over 4}\,V_{5} (a,b)}$\textbf{: }Let us consider a function $g_{V_{2} \mathunderscore V_{5} } (x)={{f}''_{V_{2} } (x)} \mathord{\left/ {\vphantom {{{f}''_{V_{2} } (x)} {{f}''_{V_{5} } (x)}}} \right. \kern-\nulldelimiterspace} {{f}''_{V_{5} } (x)}$. After simplifications, we have
\[
g_{V_{2} \mathunderscore V_{5} } (x)=\frac{8\left(
{3x^{3}+6x^{3/2}+3x+8x^{2}+x^{7/2}+6x^{5/2}+\sqrt x } \right)}{\left(
{\begin{array}{l}
 218x^{2}+174x^{3/2}+42\sqrt x +108x+ \\
 +15+174x^{5/2}+42x^{7/2}+108x^{3}+15x^{4} \\
 \end{array}} \right)},
\]
\[
\beta_{V_{2} \mathunderscore V_{5} } =g_{V_{2} \mathunderscore V_{5} }
(1)=\frac{1}{4}
\]
and
\[
\frac{1}{4}V_{5} (a,b)-V_{2} (a,b)=\frac{1}{4}U_{10}
(a,b)=\frac{1}{4}b{\kern 1pt}f_{U_{10} } \left( {\frac{a}{b}} \right),
\]
where
\begin{equation}
\label{eq39}
f_{U_{10} } (x)=\frac{\left( {\sqrt x -1} \right)^{10}}{x^{3/2}\left( {x+1}
\right)}>0,
\quad
\forall x>0,\,\,x\ne 1.
\end{equation}

\item \textbf{For }$\bf{V_{5} (a,b)\le \textstyle{1 \over 4}\,V_{9} (a,b)}$\textbf{: }Let us consider a function $g_{V_{5} \mathunderscore V_{9} } (x)={{f}''_{V_{5} } (x)} \mathord{\left/ {\vphantom {{{f}''_{V_{5} } (x)} {{f}''_{V_{9} } (x)}}} \right. \kern-\nulldelimiterspace} {{f}''_{V_{9} } (x)}$. After simplifications, we have

\[
g_{V_{4} \mathunderscore V_{2} } (x)=\frac{x\left( {\begin{array}{l}
 218x^{2}+174x^{3/2}+42\sqrt x +108x+ \\
 +15+174x^{5/2}+42x^{7/2}+108x^{3}+15x^{4} \\
 \end{array}} \right)}{4\left( {\sqrt x +1} \right)^{2}\left(
{\begin{array}{l}
 35x^{2}+32x^{3/2}+6\sqrt x +9x+60x^{5/2}+ \\
 +6x^{9/2}+32x^{7/2}+35x^{3}+9x^{4} \\
 \end{array}} \right)},
\]
\[
\beta_{V_{5} \mathunderscore V_{9} } =g_{V_{5} \mathunderscore V_{9} }
(1)=\frac{1}{4}
\]
and
\[
4V_{2} (a,b)-V_{4} (a,b)=\frac{1}{4}U_{11} (a,b)=\frac{1}{4}b{\kern
1pt}f_{U_{11} } \left( {\frac{a}{b}} \right),
\]
where
\begin{equation}
\label{eq40}
f_{U_{11} } (x)=\frac{\left( {\sqrt x +1} \right)^{2}\left( {\sqrt x -1}
\right)^{10}}{x^{2}\left( {x+1} \right)}>0,
\quad
\forall x>0,\,\,x\ne 1.
\end{equation}

\item \textbf{For }$\bf{V_{9} (a,b)\le 2\,V_{12} (a,b)}$\textbf{: }Let us consider a function $g_{V_{9} \mathunderscore V_{12} } (x)={{f}''_{V_{9} } (x)} \mathord{\left/ {\vphantom {{{f}''_{V_{9} } (x)} {{f}''_{V_{12} } (x)}}} \right. \kern-\nulldelimiterspace} {{f}''_{V_{12} } (x)}$. After simplifications, we have
\[
g_{V_{9} \mathunderscore V_{12} } (x)=\frac{2\left( {\sqrt x +1}
\right)^{2}\left( {\begin{array}{l}
 35x^{2}+32x^{3/2}+6\sqrt x +9x+60x^{5/2}+ \\
 +6x^{9/2}+32x^{7/2}+35x^{3}+9x^{4} \\
 \end{array}} \right)}{\left( {x+1} \right)^{3}\left(
{12x^{5/2}+27x^{2}+34x^{3/2}+27x+12\sqrt x } \right)},
\]
\[
\beta_{V_{9} \mathunderscore V_{12} } =g_{V_{9} \mathunderscore V_{12} }
(1)=2
\]
and
\[
2V_{12} (a,b)-V_{9} (a,b)=U_{11} (a,b).
\]
\end{enumerate}

Combining the parts 1-11, we get the proof of the inequalities (\ref{eq28}). The parts 12-14 give the proof of  (\ref{eq29}).
\end{proof}

\subsection{Third Stage}

The proof of above 14 parts give us some new measures. These are given by
\begin{equation}
\label{eq41}
U_{t} (P\vert \vert Q):=\sum\limits_{i=1}^n {q_{i} f_{U_{t} } } \left(
{\frac{p_{i} }{q_{i} }} \right),
\quad
t=1,2,...,11,
\end{equation}
where $f_{U_{t} } (x)$, $t=1,2,...,11$ are as given by (\ref{eq30})-(\ref{eq40})
respectively. In all the cases, we have $f_{U_{t} } (1)=0$, $t=1,2,...,11$.
By the application of Lemma 1.1, we can say that the above 11 measures are
convex. Here below are the second derivatives of the functions
(\ref{eq30})-(\ref{eq40}), applied frequently in next theorem.

\begin{align}
& {f}''_{U_{1} } (x)=\frac{2\left( {\sqrt x -1} \right)^{6}\left( {\sqrt x +1}
\right)^{2}\left( {x^{2}+x^{3/2}+3x+\sqrt x +1} \right)}{x^{3}(x+1)^{3}},\notag\\
& {f}''_{U_{2} } (x)=\frac{\left( {\sqrt x -1} \right)^{6}\left( {15x+26\sqrt
x +15} \right)}{4x^{3}\left( {x+1} \right)^{3}},\notag\\
& {f}''_{U_{3} } (x)=\frac{\left( {\sqrt x -1} \right)^{6}\left(
{\begin{array}{l}
 15x^{4}+54x^{7/2}+144x^{3}+246x^{5/2}+ \\
 +314x^{2}+246x^{3/2}+144x+54\sqrt x +15 \\
 \end{array}} \right)}{2x^{7/2}\left( {x+1} \right)^{3}},\notag\\
& {f}''_{U_{4} } (x)=\frac{2\left( {\sqrt x -1} \right)^{6}\left(
{\begin{array}{l}
 297x+960x^{2}+27+612x^{3/2}+ \\
 +102\sqrt x +1156x^{5/2}+612x^{7/2}+ \\
 +102x^{9/2}+27x^{5}+297x^{4}+960x^{3} \\
 \end{array}} \right)}{x^{4}\left( {x+1} \right)^{3}},\notag\\
& {f}''_{U_{5} } (x)=\frac{2\left( {\sqrt x -1} \right)^{6}\left(
{\begin{array}{l}
 73x+312x^{2}+3+180x^{3/2}+ \\
 +18\sqrt x +312x^{3}+180x^{7/2}+ \\
 +396x^{5/2}+18x^{9/2}+3x^{5}+73x^{4} \\
 \end{array}} \right)}{x^{4}\left( {x+1} \right)^{3}},\notag\\
& {f}''_{U_{6} } (x)=\frac{\left( {\sqrt x -1} \right)^{6}\left(
{\begin{array}{l}
 462x^{3/2}+462x^{5/2}+90\sqrt x +602x^{2}+ \\
 +252x+90x^{7/2}+15x^{4}+252x^{3}+15 \\
 \end{array}} \right)}{4x^{7/2}\left( {x+1} \right)^{3}},\notag\\
 & {f}''_{U_{7} } (x)=\frac{\left( {\sqrt x -1} \right)^{6}\left(
{24x^{2}+9x^{3/2}-10x+9\sqrt x +24} \right)}{4x^{4}},\notag\\
& {f}''_{U_{8} } (x)=\frac{2\left( {\sqrt x -1} \right)^{6}\left(
{3x^{2}+18x^{3/2}+28x+18\sqrt x +3} \right)}{x^{4}},\notag\\
& {f}''_{U_{9} } (x)=\frac{\left( {\sqrt x -1} \right)^{6}\left(
{12x^{2}+27x^{3/2}+34x+27\sqrt x +12} \right)}{2x^{4}},\notag\\
& {f}''_{U_{10} } (x)=\frac{\left( {\sqrt x -1} \right)^{8}\left(
{\begin{array}{l}
 15x^{3}+40x^{5/2}+77x^{2}+ \\
 +96x^{3/2}+77x+40\sqrt x +15 \\
 \end{array}} \right)}{4x^{7/2}\left( {x+1} \right)^{3}}\notag
\intertext{and}
& {f}''_{U_{11} } (x)=\frac{2\left( {\sqrt x -1} \right)^{8}\left(
{\begin{array}{l}
 3x^{4}+9x^{7/2}+22x^{3}+35x^{5/2}+ \\
 +42x^{2}+35x^{3/2}+22x+9\sqrt x +3 \\
 \end{array}} \right)}{x^{4}\left( {x+1} \right)^{3}}.\notag
\end{align}

The theorem below connects only the first nine measures. The other two are
given later.

\begin{theorem} The following inequalities hold:
\begin{equation}
\label{eq42}
U_{1} \le \textstyle{1 \over {10}}U_{6} \le \textstyle{1 \over {11}}U_{3}
\le \left\{ {\begin{array}{l}
 \textstyle{1 \over 2}U_{2} \\
 \textstyle{1 \over {56}}U_{5} \\
 \end{array}} \right\}\le \textstyle{1 \over {20}}U_{8} \le \textstyle{1
\over 8}U_{9} \le \textstyle{1 \over 2}U_{7} .
\end{equation}
\end{theorem}

\begin{proof} We shall prove the inequalities (\ref{eq42}) by parts and shall
use the same approach applied in the above theorems. Without specifying, we
shall frequently use the second derivatives ${f}''_{U_{t} } (x), t=1,2,...,11$.

\begin{enumerate}
\item \textbf{For }$\bf{U_{1} (a,b\le \textstyle{1 \over {10}}U_{6} (a,b)}$\textbf{: }Let us consider a function $g_{U_{1} \mathunderscore U_{6} } (x)={{f}''_{U_{1} } (x)} \mathord{\left/ {\vphantom {{{f}''_{U_{1} } (x)} {{f}''_{U_{6} } (x)}}} \right. \kern-\nulldelimiterspace} {{f}''_{U_{6} } (x)}$. After simplifications, we have
\[
g_{U_{1} \mathunderscore U_{6} } (x)=\frac{x\left(
{3x+3x^{3}+8x^{2}+x^{7/2}+6x^{3/2}+6x^{5/2}+\sqrt x } \right)}{\left(
{\begin{array}{l}
 462x^{3/2}+462x^{5/2}+90\sqrt x +602x^{2} \\
 +252x+90x^{7/2}+15x^{4}+252x^{3}+15 \\
 \end{array}} \right)},
\]
\[
\beta_{U_{1} \mathunderscore U_{6} } =g_{U_{1} \mathunderscore U_{6} }
(1)=\frac{1}{10}
\]
and
\[
\frac{1}{10}U_{6} (a,b)-U_{1} (a,b)=\frac{1}{10}U_{10} (a,b).
\]

\item \textbf{For }$\bf{U_{6} (a,b)\le \textstyle{{10} \over {11}}U_{3} (a,b)}$\textbf{: }Let us consider a function $g_{U_{6} \mathunderscore U_{3} } (x)={{f}''_{U_{6} } (x)} \mathord{\left/ {\vphantom {{{f}''_{U_{6} } (x)} {{f}''_{U_{3} } (x)}}} \right. \kern-\nulldelimiterspace} {{f}''_{U_{3} } (x)}$. After simplifications, we have
\[
g_{U_{6} \mathunderscore U_{3} } (x)=\frac{\left( {\begin{array}{l}
 462x^{3/2}+462x^{5/2}+90\sqrt x +602x^{2}+ \\
 +252x+90x^{7/2}+15x^{4}+252x^{3}+15 \\
 \end{array}} \right)}{2\left( {\begin{array}{l}
 246x^{3/2}+246x^{5/2}+54\sqrt x +314x^{2}+ \\
 +144x+54x^{7/2}+144x^{3}+15x^{4}+15 \\
 \end{array}} \right)},
\]
\[
\beta_{U_{6} \mathunderscore U_{3} } =g_{U_{6} \mathunderscore U_{3} }
(1)=\frac{10}{11}
\]
and
\[
\frac{10}{11}U_{3} (a,b)-U_{6} (a,b)=\frac{9}{11}U_{10} (a,b).
\]

\item \textbf{For }$\bf{U_{3} (a,b)\le \textstyle{{11} \over {56}}U_{5} (a,b)}$\textbf{: }Let us consider a function $g_{U_{3} \mathunderscore U_{5} } (x)={{f}''_{U_{3} } (x)} \mathord{\left/ {\vphantom {{{f}''_{U_{3} } (x)} {{f}''_{U_{5} } (x)}}} \right. \kern-\nulldelimiterspace} {{f}''_{U_{5} } (x)}$. After simplifications, we have
\[
g_{U_{3} \mathunderscore U_{5} } (x)=\frac{x\left( {\begin{array}{l}
 246x^{3/2}+246x^{5/2}+54\sqrt x +314x^{2}+ \\
 +144x+54x^{7/2}+144x^{3}+15x^{4}+15 \\
 \end{array}} \right)}{4\left( {\begin{array}{l}
 73x^{3/2}+312x^{5/2}+3\sqrt x +180x^{2}+73x^{9/2}+ \\
 +312x^{7/2}+180x^{4}+396x^{3}+3x^{11/2}+18x^{5}+18x \\
 \end{array}} \right)},
\]
\[
\beta_{U_{3} \mathunderscore U_{5} } =g_{U_{3} \mathunderscore U_{5} }
(1)=\frac{11}{56}
\]
and
\[
\frac{11}{56}U_{5} (a,b)-U_{3} (a,b)=\frac{1}{56}U_{12}
(a,b)=\frac{1}{56}b{\kern 1pt}f_{U_{12} } \left( {\frac{a}{b}} \right),
\]
where
\begin{equation}
\label{eq43}
f_{U_{12} } (x)=\frac{\left( {\sqrt x -1} \right)^{10}\left( {11x-2\sqrt x
+11} \right)}{x^{2}\left( {x+1} \right)}>0,
\quad
\forall x>0,\,\,x\ne 1.
\end{equation}

\item \textbf{For }$\bf{U_{3} (a,b)\le \textstyle{{11} \over 2}U_{2} (a,b)}$\textbf{: }Let us consider a function $g_{U_{3} \mathunderscore U_{2} } (x)={{f}''_{U_{3} } (x)} \mathord{\left/ {\vphantom {{{f}''_{U_{3} } (x)} {{f}''_{U_{2} } (x)}}} \right. \kern-\nulldelimiterspace} {{f}''_{U_{2} } (x)}$. After simplifications, we have
\[
g_{U_{3} \mathunderscore U_{2} } (x)=\frac{2\left( {\begin{array}{l}
 246x^{3/2}+246x^{5/2}+54\sqrt x +314x^{2}+ \\
 +144x+54x^{7/2}+144x^{3}+15x^{4}+15 \\
 \end{array}} \right)}{\left( {15x+26\sqrt x +15} \right)\left( {x+1}
\right)^{3}},
\]
\[
\beta_{U_{3} \mathunderscore U_{2} } =g_{U_{3} \mathunderscore U_{2} }
(1)=\frac{11}{2}
\]
and
\[
\frac{11}{2}U_{2} (a,b)-U_{3} (a,b)=\frac{7}{2}U_{10} (a,b).
\]

\item \textbf{For }$\bf{U_{2} (a,b)\le \textstyle{1 \over {10}}U_{8} (a,b)}$\textbf{: }Let us consider a function $g_{U_{2} \mathunderscore U_{8} } (x)={{f}''_{U_{2} } (x)} \mathord{\left/ {\vphantom {{{f}''_{U_{2} } (x)} {{f}''_{U_{8} } (x)}}} \right. \kern-\nulldelimiterspace} {{f}''_{U_{8} } (x)}$. After simplifications, we have
\[
g_{U_{2} \mathunderscore U_{8} } (x)=\frac{\sqrt x \left( {15x+26\sqrt x
+15} \right)}{8\left( {3x^{2}+28x+18x^{3/2}+18\sqrt x +3} \right)},
\]
This gives $\beta_{U_{2} \mathunderscore U_{8} } =g_{U_{2} \mathunderscore
U_{8} } (1)=\frac{1}{10}$. Let us consider now,

and
\[
\frac{1}{10}U_{8} (a,b)-U_{2} (a,b)=\frac{1}{10}U_{13}
(a,b)=\frac{1}{10}b{\kern 1pt}f_{U_{13} } \left( {\frac{a}{b}} \right),
\]
where
\begin{equation}
\label{eq44}
f_{U_{13} } (x)=\frac{\left( {\sqrt x -1} \right)^{10}}{x^{2}}>0,
\quad
\forall x>0,\,\,x\ne 1.
\end{equation}

\item \textbf{For }$\bf{U_{5} (a,b)\le \textstyle{{14} \over 5}U_{8} (a,b)}$\textbf{: }Let us consider a function $g_{U_{5} \mathunderscore U_{8} } (x)={{f}''_{U_{5} } (x)} \mathord{\left/ {\vphantom {{{f}''_{U_{5} } (x)} {{f}''_{U_{8} } (x)}}} \right. \kern-\nulldelimiterspace} {{f}''_{U_{8} } (x)}$. After simplifications, we have
\[
g_{U_{5} \mathunderscore U_{8} } (x)=\frac{\left( {\begin{array}{l}
 73x^{3/2}+312x^{5/2}+3\sqrt x +180x^{2}+18x+73x^{9/2} \\
 +312x^{7/2}+180x^{4}+396x^{3}+3x^{11/2}+18x^{5} \\
 \end{array}} \right)}{\left( {3x^{5/2}+28x^{3/2}+18x^{2}+18x+3\sqrt x }
\right)\left( {x+1} \right)^{3}},
\]
\[
\beta_{U_{5} \mathunderscore U_{8} } =g_{U_{5} \mathunderscore U_{8} }
(1)=\frac{14}{5}
\]
and
\[
\frac{14}{5}U_{8} (a,b)-U_{5} (a,b)=\frac{9}{5}U_{14}
(a,b)=\frac{9}{5}b{\kern 1pt}f_{U_{14} } \left( {\frac{a}{b}} \right),
\]
where
\begin{equation}
\label{eq45}
f_{U_{14} } (x)=\frac{\left( {\sqrt x -1} \right)^{10}\left( {x+10\sqrt x
+1} \right)}{x^{2}\left( {x+1} \right)}>0,
\quad
\forall x>0,\,\,x\ne 1.
\end{equation}

\item \textbf{For }$\bf{U_{8} (a,b)\le \textstyle{5 \over 2}U_{9} (a,b)}$\textbf{: }Let us consider a function $g_{U_{8} \mathunderscore U_{9} } (x)={{f}''_{U_{8} } (x)} \mathord{\left/ {\vphantom {{{f}''_{U_{8} } (x)} {{f}''_{U_{9} } (x)}}} \right. \kern-\nulldelimiterspace} {{f}''_{U_{9} } (x)}$. After simplifications, we have
\[
g_{U_{8} \mathunderscore U_{9} } (x)=\frac{4\left(
{3x^{2}+18x^{3/2}+28x+18\sqrt x +3} \right)}{12x^{2}+27x^{3/2}+34x+27\sqrt x
+12},
\]
\[
\beta_{U_{8} \mathunderscore U_{9} } =g_{U_{8} \mathunderscore U_{9} }
(1)=\frac{5}{2}
\]
and
\[
\frac{5}{2}U_{9} (a,b)-U_{8} (a,b)=\frac{3}{2}U_{13} (a,b),
\]

\item \textbf{For }$\bf{U_{9} (a,b)\le 4\,U_{7} (a,b)}$\textbf{: }Let us consider a function $g_{U_{9} \mathunderscore U_{7} } (x)={{f}''_{U_{9} } (x)} \mathord{\left/ {\vphantom {{{f}''_{U_{9} } (x)} {{f}''_{U_{7} } (x)}}} \right. \kern-\nulldelimiterspace} {{f}''_{U_{7} } (x)}$. After simplifications, we have
\[
g_{U_{9} \mathunderscore U_{7} } (x)=\frac{2\left(
{12x^{2}+27x^{3/2}+34x+27\sqrt x +12} \right)}{24x^{2}+9x^{3/2}-10x+9\sqrt x
+24},
\]
\[
\beta_{U_{9} \mathunderscore U_{7} } =g_{U_{9} \mathunderscore U_{7} }
(1)=4
\]
and
\[
4U_{7} (a,b)-U_{9} (a,b)=3U_{13} (a,b).
\]
\end{enumerate}
Combining the parts 1-8, we get the proof of the inequalities (\ref{eq42}).
\end{proof}

\subsection{Forth Stage}

Still, we have more measures to compares, i.e., $U_{10} $ to $U_{14} $. This
comparison is given in the theorem below. Here below are the second
derivatives of the functions given by (\ref{eq43})-(\ref{eq45}).
\begin{align}
& {f}''_{U_{12} } (x)=\frac{\left( {\sqrt x -1} \right)^{8}\left(
{\begin{array}{l}
 15x^{3}+40x^{5/2}+77x^{2}+ \\
 +96x^{3/2}+77x+40\sqrt x +15 \\
 \end{array}} \right)}{4x^{7/2}\left( {x+1} \right)^{3}},\notag\\
& {f}''_{U_{13} } (x)=\frac{\left( {\sqrt x -1} \right)^{8}\left(
{\begin{array}{l}
 15x^{3}+40x^{5/2}+77x^{2}+ \\
 +96x^{3/2}+77x+40\sqrt x +15 \\
 \end{array}} \right)}{4x^{7/2}\left( {x+1} \right)^{3}},\notag
\intertext{and}
& {f}''_{U_{14} } (x)=\frac{2\left( {\sqrt x -1} \right)^{8}\left(
{\begin{array}{l}
 3x^{4}+9x^{7/2}+22x^{3}+35x^{5/2}+ \\
 +42x^{2}+35x^{3/2}+22x+9\sqrt x +3 \\
 \end{array}} \right)}{x^{4}\left( {x+1} \right)^{3}}.\notag
\end{align}

\begin{theorem} The following inequalities hold:
\begin{equation}
\label{eq46}
U_{10} \le \textstyle{1 \over {12}}U_{14} \le \textstyle{1 \over 4}U_{11}
\le \textstyle{1 \over 2}U_{13} \le \textstyle{1 \over {20}}U_{12} .
\end{equation}
\end{theorem}

\begin{proof} We shall prove the above theorem by parts.
\begin{enumerate}
\item \textbf{For }$\bf{U_{10} (a,b)\le \textstyle{1 \over {12}}\,U_{14} (a,b)}$\textbf{: }Let us consider a function $g_{U_{10} \mathunderscore U_{14} } (x)={{f}''_{U_{10} } (x)} \mathord{\left/ {\vphantom {{{f}''_{U_{10} } (x)} {{f}''_{U_{14} } (x)}}} \right. \kern-\nulldelimiterspace} {{f}''_{U_{14} } (x)}$. After simplifications, we have
\[
g_{U_{10} \mathunderscore U_{14} } (x)=\frac{x\left( {\begin{array}{l}
 15+77x^{2}+77x+40x^{5/2}+ \\
 +15x^{3}+40\sqrt x +96x^{3/2} \\
 \end{array}} \right)}{8\left( {\begin{array}{l}
 62x^{3/2}+138x^{5/2}+3\sqrt x +112x^{2}+ \\
 24x+3x^{9/2}+62x^{7/2}+24x^{4}+112x^{3} \\
 \end{array}} \right)},
\]
\[
\beta_{U_{10} \mathunderscore U_{14} } =g_{U_{10} \mathunderscore U_{14} }
(1)=\frac{1}{12}
\]
and
\[
\frac{1}{12}U_{14} (a,b)-U_{10} (a,b)=\frac{1}{12}U_{15}
(a,b)=\frac{1}{12}b{\kern 1pt}f_{U_{15} } \left( {\frac{a}{b}} \right).
\]
where
\begin{equation}
\label{eq47}
f_{U_{15} } (x)=\frac{\left( {\sqrt x -1} \right)^{12}}{x^{2}\left( {x+1}
\right)}>0,
\quad
\forall x>0,\,\,x\ne 1.
\end{equation}

\item \textbf{For }$\bf{U_{14} (a,b)\le 3\,U_{11} (a,b)}$\textbf{: }Let us consider a function $g_{U_{14} \mathunderscore U_{11} } (x)={{f}''_{U_{14} } (x)} \mathord{\left/ {\vphantom {{{f}''_{U_{14} } (x)} {{f}''_{U_{11} } (x)}}} \right. \kern-\nulldelimiterspace} {{f}''_{U_{11} } (x)}$. After simplifications, we have
\[
g_{U_{14} \mathunderscore U_{11} } (x)=\frac{\left( {\begin{array}{l}
 62x^{3/2}+138x^{5/2}+3\sqrt x +112x^{2}+ \\
 +24x+3x^{9/2}+62x^{7/2}+24x^{4}+112x^{3} \\
 \end{array}} \right)}{\left( {\begin{array}{l}
 22x^{3/2}+42x^{5/2}+3\sqrt x +35x^{2}+ \\
 +9x+3x^{9/2}+22x^{7/2}+9x^{4}+35x^{3} \\
 \end{array}} \right)},
\]
\[
\beta_{U_{14} \mathunderscore U_{11} } =g_{U_{14} \mathunderscore U_{11} }
(1)=3
\]
and
\[
3U_{11} (a,b)-U_{14} (a,b)=2U_{15} (a,b).
\]

\item \textbf{For }$\bf{U_{11} (a,b)\le 2\,U_{13} (a,b)}$\textbf{: }Let us consider a function $g_{U_{11} \mathunderscore U_{13} } (x)={{f}''_{U_{11} } (x)} \mathord{\left/ {\vphantom {{{f}''_{U_{11} } (x)} {{f}''_{U_{13} } (x)}}} \right. \kern-\nulldelimiterspace} {{f}''_{U_{13} } (x)}$. After simplifications, we have
\[
g_{U_{11} \mathunderscore U_{13} } (x)=\frac{4\left( {\begin{array}{l}
 22x+42x^{2}+3+35x^{3/2}+ \\
 +9\sqrt x +3x^{4}+22x^{3}+9x^{7/2}+35x^{5/2} \\
 \end{array}} \right)}{3\left( {4x+7\sqrt x +4} \right)\left( {x+1}
\right)^{3}},
\]
\[
\beta_{U_{11} \mathunderscore U_{13} } =g_{U_{11} \mathunderscore U_{13} }
(1)=2
\]
and
\[
2U_{13} (a,b)-U_{11} (a,b)=U_{15} (a,b).
\]

\item \textbf{For }$\bf{U_{13} (a,b)\le \textstyle{1 \over {10}}U_{12} (a,b)}$\textbf{: }Let us consider a function $g_{U_{13} \mathunderscore U_{12} } (x)={{f}''_{U_{13} } (x)} \mathord{\left/ {\vphantom {{{f}''_{U_{13} } (x)} {{f}''_{U_{12} } (x)}}} \right. \kern-\nulldelimiterspace} {{f}''_{U_{12} } (x)}$. After simplifications, we have
\[
g_{U_{13} \mathunderscore U_{12} } (x)=\frac{3\left( {4x+7\sqrt x +4}
\right)\left( {x+1} \right)^{3}}{4\left( {\begin{array}{l}
 122x+174x^{2}+33+154x^{3/2}+54\sqrt x + \\
 +33x^{4}+122x^{3}+54x^{7/2}+154x^{5/2} \\
 \end{array}} \right)},
\]
\[
\beta_{U_{13} \mathunderscore U_{12} } =g_{U_{13} \mathunderscore U_{12} }
(1)=\frac{1}{10}
\]
and
\[
\frac{1}{10}U_{12} (a,b)-U_{13} (a,b)=\frac{1}{10}U_{15} (a,b).
\]
\end{enumerate}
\end{proof}

\begin{remark} Interestingly, in all the four cases we are left with only a
single measure, i.e., $U_{15} (a,b)$ given by
\begin{equation}
\label{eq48}
U_{15} (a,b)=\frac{\left( {\sqrt a -\sqrt b } \right)^{12}}{\left( {ab}
\right)^{2}\left( {a+b} \right)},
\quad
a,b>0,
\quad
a\ne b.
\end{equation}
\end{remark}

\subsection{Equivalent Expressions}

The measures appearing in the proof of the Theorems 2.2-2.4 can be written
in terms of the measures appearing in the inequalities (\ref{eq8}). Here below are
equivalent versions of these measures.

\bigskip
\noindent \textbf{$\bullet$ Measures appearing in Theorem 2.2}. We can write
\begin{align}
V_{1} & = {K+26\Delta -48D_{CN} } \notag\\
&\hspace{30pt} = {K+30D_{CN} -26D_{CG} } \notag\\
& \hspace{30pt}= {K+14D_{CG} -30D_{RG} } \notag\\
& \hspace{30pt}= {K+12D_{RG} -28hD_{RG} },\notag\\
V_{2} & =\Psi +64h-4\Delta -8K,\notag\\
V_{3} & = {\Psi +108\Delta -192D_{CN} } \notag\\
& \hspace{30pt} = {\Psi +132D_{CN} -108D_{CG} } \notag\\
& \hspace{30pt}= {\Psi +68D_{CG} -132D_{RG} } \notag\\
& \hspace{30pt}= {\Psi +72D_{RG} -136h},\notag\\
V_{4} & =\Psi +32h-6K,\notag\\
V_{5} & =2\left( {F+6K-4\Delta -3\Psi } \right),\notag\\
V_{6} & =2\left( {F+164\Delta -288D_{CN} } \right)\notag\\
& \hspace{30pt}= 2\left( {F+204D_{CN} -164D_{CG} } \right)\notag\\
& \hspace{30pt}=2 \left( {F+108D_{CG} -204D_{RG} } \right) \notag\\
& \hspace{30pt}=2\left(
{F+120D_{RG} -216h} \right),\notag\\
V_{7} & =2\left( {F-10K+64h} \right),\notag\\
V_{8} & =2\left( {F+2K-2\Psi } \right),\notag\\
V_{9} & =L+16K-16\Delta -8F,\notag\\
V_{10} & = {L+880\Delta -1536D_{CN} } \notag\\
& \hspace{30pt}={L+1104D_{CN} -880D_{CG} } \notag\\
& \hspace{30pt} = {L+592D_{CG} -1104D_{RG} } \notag\\
& \hspace{30pt}= {L+6724D_{RG}-1184h} \notag\\
V_{11} & =L+384h-56K,\notag\\
V_{12} & =L+12\Psi -8K-12F,\notag\\
V_{13} & =L+16K-12\Psi,\notag\\
V_{14} & =L+4\Psi -8F.\notag
\end{align}

\bigskip
\noindent \textbf{$\bullet$ Measures appearing in Theorem 2.3 and 2.4}. We can write
\begin{align}
U_{1} & =\Psi +192D_{CN} -100\Delta -8K,\notag\\
U_{2} & =2\left( {F+14K-64h-4\Psi } \right),\notag\\
U_{3} & =4F+576D_{CN} -316\Delta -9\Psi,\notag\\
U_{4} & =9L+4608D_{CN} -2576\Delta -64F,\notag\\
U_{5} & =\textstyle{1 \over 7}\left( {7L+6144h+13824D_{CN} -896K-7920\Delta }
\right),\notag\\
U_{6} & =\textstyle{2 \over 5}\left( {5F+576h+1152D_{CN} -90K-656\Delta }
\right),\notag\\
U_{7} & =L+384h+36\Psi -92K-18F,\notag\\
U_{8} & =L+160K-36\Psi -768h,\notag\\
U_{9} & =L+12\Psi -8K-12F,\notag\\
U_{10} & =2\left( {F+22K+4\Delta -5\Psi -128h} \right),\notag\\
U_{11} &=L+16\Delta +24\Psi -16F-32K,\notag\\
U_{12} & =\textstyle{1 \over 7}\left( {77L-9856K+67584h-73728D_{CN}
+36752\Delta -1568F+3528\Psi } \right),\notag\\
U_{13} & =L+44\Psi -120K+512h-20F\notag\\
U_{14} & =\textstyle{1 \over 7}\left( {7L-392\Psi +2240K-11776h-7680D_{CN}
+4400\Delta } \right),\notag\\
U_{15} & =\textstyle{1 \over 7}\left( {7L+448\Delta -1456K+9728h-7680D_{CN}
+3728\Delta -168F} \right).\notag
\end{align}

\section{Generating Divergence Measures}

Some of the measures given in Section 2 can be written in generating forms.
Let us see below these generating measures.

\subsection{First Generalization of Triangular Discrimination}

For all $(a,b)\in {\rm R}_{+}^{2} $, let consider the following measures
\begin{equation}
\label{eq49}
\Delta_{t}^{1} (a,b)=\frac{\left( {a-b} \right)^{2}\left( {\sqrt a -\sqrt b
} \right)^{2t}}{\left( {a+b} \right)\left( {\sqrt {ab} } \right)^{t}},\quad
t=0,1,2,3,...
\end{equation}
In particular, we have
\begin{align}
& \Delta_{0}^{1} =\Delta =\frac{\left( {a-b} \right)^{2}}{\left( {a+b}
\right)},\notag\\
& \Delta_{1}^{1} =D_{W_{6} W_{1} }^{15} =K-2\Delta =\frac{\left( {a-b}
\right)^{2}\left( {\sqrt a -\sqrt b } \right)^{2}}{\left( {a+b} \right)\sqrt
{ab} },\notag\\
& \Delta_{2}^{1} =\Psi -4K+4\Delta =\frac{\left( {a-b} \right)^{2}\left(
{\sqrt a -\sqrt b } \right)^{4}}{ab\left( {a+b} \right)}\notag
\intertext{and}
& \Delta_{3}^{1} =V_{5}
=2\left( {F+6K-4\Delta -3\Psi } \right)=\frac{\left( {a-b} \right)^{2}\left(
{\sqrt a -\sqrt b } \right)^{6}}{\left( {ab} \right)^{3/2}\left( {a+b}
\right)}.\notag
\end{align}

The expression (\ref{eq49}) gives first generalization of the measure $\Delta
(a,b)$. Let us prove now its convexity. We can write $\Delta_{t}^{1}
(a,b)=bf_{\Delta_{t}^{1} } (a/b)$, $t\in {\rm N}$, where
\begin{equation}
\label{eq50}
f_{\Delta_{t}^{1} } (x)=\frac{\left( {x-1} \right)^{2}\left( {\sqrt x -1}
\right)^{2t}}{\left( {x+1} \right)\left( {\sqrt x } \right)^{t}}.
\end{equation}

The second order derivative of the function $f_{\Delta_{t}^{1} } (x)$ is
given by
\[
{f}''_{\Delta_{t}^{1} } (x)=\frac{\left( {\sqrt x -1}
\right)^{2t}}{4x^{2}\left( {x+1} \right)^{3}\left( {\sqrt x }
\right)^{t}}\times A_{1} (x,t),
\]
where
\[
A_{1} (x,t)=\left( {\begin{array}{l}
 t\left( {t+2} \right)\left( {x^{4}+1} \right)+2t\left( {2t+1} \right)\sqrt
x \left( {x^{3}+1} \right)+ \\
 +4t\left( {2t+3} \right)x\left( {x^{2}+1} \right)+4\left( {7t^{2}+10t+16}
\right)x^{2} \\
 +2t\left( {6t+11} \right)x^{3/2}\left( {x+1} \right) \\
 \end{array}} \right).
\]

For all $t\ge 0$, $x>0$, $x\ne 1$, we have ${f}''_{\Delta_{t}^{1} } (x)>0$.
Also we have $f_{\Delta_{t}^{1} } (1)=0$. In view of Lemma 1.1, the measure
$\Delta_{t}^{1} (a,b)$ is convex for all $(a,b)\in {\rm R}_{+}^{2} $, $t\in
{\rm N}$.

\bigskip
Now, we shall present exponential representation of the measure (\ref{eq49}) based
on the function given by (\ref{eq50}). Let us consider a linear combination of
convex functions,
\[
f_{\Delta^{1}} (x)=a_{0} f_{\Delta_{0}^{1} } (x)+a_{1} f_{\Delta_{1}^{1}
} (x)+a_{2} f_{\Delta_{2}^{1} } (x)+a_{3} f_{\Delta_{3}^{1} } (x)+...
\]
i.e.,
\[
f_{\Delta^{1}} (x)=a_{0} \frac{(x-1)^{2}}{x+1}+a_{1} \frac{\left( {x-1}
\right)^{2}\left( {\sqrt x -1} \right)^{2}}{\left( {x+1} \right)\sqrt x
}+a_{2} \frac{\left( {x-1} \right)^{2}\left( {\sqrt x -1}
\right)^{4}}{x\left( {x+1} \right)}+
\]
\[
+a_{3} \frac{\left( {x-1} \right)^{2}\left( {\sqrt x -1} \right)^{6}}{\left(
x \right)^{3/2}\left( {x+1} \right)}+...,
\]
where $a_{0} ,a_{1} ,\,a_{2} ,\,a_{3} ,...$are the constants. For simplicity
let us choose,
\[
a_{0} =\frac{1}{0!},\,a_{1} =\frac{1}{1!},\,a_{2} =\frac{1}{2!},\,a_{3}
=\frac{1}{3!},...
\]

Thus we have
\begin{align}
f_{\Delta^{1}}
(x)& =\frac{1}{0!}\frac{(x-1)^{2}}{x+1}+\frac{1}{1!}\frac{\left( {x-1}
\right)^{2}\left( {\sqrt x -1} \right)^{2}}{\left( {x+1} \right)\sqrt x
}+\frac{1}{2!}\frac{\left( {x-1} \right)^{2}\left( {\sqrt x -1}
\right)^{4}}{x\left( {x+1} \right)}+\notag\\
& \hspace{40pt} +\frac{1}{3!}\frac{\left( {x-1} \right)^{2}\left( {\sqrt x -1}
\right)^{6}}{\left( x \right)^{3/2}\left( {x+1} \right)}+...\notag\\
& =\frac{(x-1)^{2}}{x+1}\left[ {\frac{1}{0!}+\frac{1}{1!}\left( {\frac{(\sqrt
x -1)^{2}}{\sqrt x }} \right)^{1}+\frac{1}{2!}\left( {\frac{(\sqrt x
-1)^{2}}{\sqrt x }} \right)^{2}+\frac{1}{3!}\left( {\frac{(\sqrt x
-1)^{2}}{\sqrt x }} \right)^{3}+...} \right].\notag
\end{align}

This gives us
\begin{equation}
\label{eq51}
f_{\Delta^{1}} (x)=\frac{(x-1)^{2}}{x+1}\exp \left( {\frac{(x-1)^{2}}{\sqrt
x }} \right).
\end{equation}

As a consequence of (\ref{eq51}), we have the following exponential representation of triangular
discrimination
\begin{equation}
\label{eq52}
E_{\Delta^{1}} (a,b)=b\,f_{\Delta^{1}} (a/b)=\frac{(a-b)^{2}}{a+b}\exp
\left( {\frac{(a-b)^{2}}{\sqrt {ab} }} \right),\quad (a,b)\in {\rm
R}_{+}^{2} .
\end{equation}

\subsection{Second Generalization of Triangular Discrimination}

For all $(a,b)\in {\rm R}_{+}^{2} $, let consider the following measures
\begin{equation}
\label{eq53}
\Delta_{t}^{2} (a,b)=\frac{\left( {a-b} \right)^{2(t+1)}}{\left( {a+b}
\right)\left( {ab} \right)^{t}},\quad t=0,1,2,3,...
\end{equation}

In particular, we have
\begin{align}
& \Delta_{0}^{2} =\Delta =\frac{\left( {a-b} \right)^{2}}{a+b}\notag
\intertext{and}
& \Delta_{1}^{2} =2D_{W_{7} W_{1} }^{21} =\Psi -4\Delta =\frac{\left( {a-b}
\right)^{4}}{ab\left( {a+b} \right)}.\notag
\end{align}

The expression (\ref{eq53}) gives us a second generalization of the measure $\Delta
(a,b)$. Let us prove now its convexity. We can write $\Delta_{t}^{2}
(a,b)=b\,f_{\Delta_{t}^{2} } (a/b)$, $t\in {\rm N}$, where
\[
f_{\Delta_{t}^{2} } (x)=\frac{\left( {\sqrt x -1} \right)^{2(t+1)}}{\left(
{x+1} \right)x^{t}}.
\]

The second order derivative of the function $f_{\Delta_{t}^{2} } (x)$ is
given by
\[
{f}''_{\Delta_{t}^{2} } (x)=\frac{\left( {x-1} \right)^{2t}}{\left( {x+1}
\right)^{3}x^{t+2}}\times A_{2} (x,t),
\]
where
\[
A_{2} (x,t)=t\left( {t+1} \right)\left( {x^{4}+1} \right)+2t\left( {2t+3}
\right)x\left( {x^{2}+1} \right)+2\left( {3t^{2}+5t+4} \right)x^{2}.
\]

For all $t\ge 0$, $x>0$, $x\ne 1$, we have ${f}''_{\Delta_{t}^{2} } (x)>0$.
Also we have $f_{\Delta_{t}^{2} } (1)=0$. In view of Lemma 1.1, the measure
$\Delta_{t}^{2} (a,b)$ is convex for all $(a,b)\in {\rm R}_{+}^{2} $, $t\in
{\rm N}$.

\bigskip
Following similar lines of (\ref{eq51}) and (\ref{eq52}), the exponential
representation of the measure $\Delta_{t}^{2} (a,b)$ is given by
\[
E_{\Delta^{2}} (a,b)=\frac{(a-b)^{2}}{a+b}\exp \left(
{\frac{(a-b)^{2}}{ab}} \right),\quad (a,b)\in {\rm R}_{+}^{2} .
\]

\subsection{First Generalization of the Measure $\bf{K(a,b)}$}

For all $(a,b)\in {\rm R}_{+}^{2} $, let consider the following measures
\begin{equation}
\label{eq54}
K_{t}^{1} (a,b)=\frac{\left( {a-b} \right)^{2}\left( {\sqrt a -\sqrt b }
\right)^{2t}}{\left( {\sqrt {ab} } \right)^{t+1}},\quad t=0,1,2,3,...
\end{equation}

In particular, we have
\begin{align}
& K_{0}^{1} =K=\frac{\left( {a-b} \right)^{2}}{\sqrt {ab} },\notag\\
& K_{1}^{1} =2D_{W_{7} W_{6} }^{16} =\Psi -2K=\frac{\left( {a-b}
\right)^{2}\left( {\sqrt a -\sqrt b } \right)^{2}}{ab},\notag\\
& K_{2}^{1} =
V_{8} =2\left( {F+2K-2\Psi } \right)=\frac{\left( {a-b} \right)^{2}\left(
{\sqrt a -\sqrt b } \right)^{4}}{\left( {ab} \right)^{3/2}}\notag
\intertext{and}
& K_{3}^{1} =
V_{12} =L+12\Psi -8K-12F=\frac{\left( {a-b} \right)^{2}\left( {\sqrt a
-\sqrt b } \right)^{6}}{\left( {ab} \right)^{2}}.
\end{align}

The expression (\ref{eq55}) gives first parametric generalization the measure
$K(a,b)$ given by (\ref{eq3}). Let us prove now its convexity. We can write
$K_{t}^{1} (a,b)=b\,f_{K_{t}^{1} } (a/b)$, $t\in {\rm N}$, where
\[
f_{K_{t}^{1} } (x)=\frac{\left( {x-1} \right)^{2}\left( {\sqrt x -1}
\right)^{2t}}{\left( {\sqrt x } \right)^{t+1}}.
\]

The second order derivative of the function $f_{K_{t}^{1} } (x)$ is given by
\[
{f}''_{K_{t}^{1} } (x)=\frac{\left( {\sqrt x -1} \right)^{2t}}{4x^{2}\left(
{\sqrt x } \right)^{t+1}}\times A_{3} (x,t),
\]
where
\[
A_{3} (x,t)=\left( {t+1} \right)\left( {t+3} \right)\left( {x^{2}+1}
\right)+2t\left( {2t+3} \right)\sqrt x \left( {x+1} \right)+2\left(
{3t^{2}+2t+1} \right)x.
\]

\bigskip
For all $t\ge 0$, $x>0$, $x\ne 1$, we have ${f}''_{K_{t}^{1} } (x)>0$. Also
we have $f_{K_{t}^{1} } (1)=0$. In view of Lemma 1.1, the measure $K_{t}^{1}
(a,b)$ is convex for all $(a,b)\in {\rm R}_{+}^{2} $, $t\in {\rm N}$.

\bigskip
Following similar lines of (\ref{eq51}) and (\ref{eq52}), the exponential
representation of the measure $K_{t}^{1} (a,b)$ is given by
\[
E_{K^{1}} (a,b)=\frac{\left( {\sqrt a -\sqrt b } \right)^{2}}{\sqrt {ab}
}\exp \left( {\frac{\left( {a-b} \right)^{2}}{\sqrt {ab} }} \right),\quad
(a,b)\in {\rm R}_{+}^{2} .
\]

\subsection{Second Generalization of the Measure $\bf{K(a,b)}$}

For all $(a,b)\in {\rm R}_{+}^{2} $, let consider the following measures
\begin{equation}
\label{eq55}
K_{t}^{2} (a,b)=\frac{\left( {a-b} \right)^{2(t+1)}}{\left( {\sqrt {ab} }
\right)^{2t+1}},\, t=0,1,2,3,...
\end{equation}

In particular, we have
\begin{align}
& K_{0}^{2} =K=\frac{\left( {a-b} \right)^{2}}{\sqrt {ab} }\notag
\intertext{and}
& K_{1}^{2} =4D_{W_{8} W_{6} }^{23} =F-2K=\frac{\left( {a-b}
\right)^{4}}{\left( {ab} \right)^{3/2}}.\notag
\end{align}

The expression (\ref{eq55}) gives second generalization the measure $K(a,b)$ given
by (\ref{eq3}). Let us prove now its convexity. We can write $K_{t}^{2}
(a,b)=b\,f_{K_{t}^{2} } (a/b)$, $t\in {\rm N}$, where
\[
f_{K_{t}^{2} } (x)=\frac{\left( {x-1} \right)^{2(t+1)}}{\left( {\sqrt x }
\right)^{2t+1}}.
\]

The second order derivative of the function $f_{K_{t}^{2} } (x)$ is given by
\[
{f}''_{K_{t}^{2} } (x)=\frac{\left( {x-1} \right)^{2t}}{4x^{2}\left( {\sqrt
x } \right)^{2t+1}}\times A_{4} (x,t),
\]
where
\[
A_{4} (x,t)=\left( {2t+1} \right)\left[ {2tx^{2}+3x^{2}+2\left( {2t+1}
\right)x+2t+3} \right].
\]

\bigskip
For all $t\ge 0$, $x>0$, $x\ne 1$, we have ${f}''_{K_{t}^{2} } (x)>0$. Also
we have $f_{K_{t}^{2} } (1)=0$. In view of Lemma 1.1, the measure $K_{t}^{2}
(a,b)$ is convex for all $(a,b)\in {\rm R}_{+}^{2} $, $t\in {\rm N}$.

\bigskip
Following similar lines of (\ref{eq51}) and (\ref{eq52}), the exponential representation of the measure
$K_{t}^{2} (a,b)$ is given by
\[
E_{K^{2}} (a,b)=\frac{(a-b)^{2}}{\sqrt {ab} }\exp \left(
{\frac{(a-b)^{2}}{ab}} \right),\quad (a,b)\in {\rm R}_{+}^{2} .
\]

\subsection{Generalization of Hellingar's Discrimination}

For all $(a,b)\in {\rm R}_{+}^{2} $, let consider the following measures
\begin{equation}
\label{eq56}
h_{t} (a,b)=\frac{\left( {\sqrt a -\sqrt b } \right)^{2(t+1)}}{\left( {\sqrt
{ab} } \right)^{t}},\;t\in {\rm N}
\end{equation}

In particular, we have
\begin{align}
& h_{0} =2h=\left( {\sqrt a -\sqrt b } \right)^{2},\notag\\
& h_{1} =D_{W_{6} W_{5} }^{11} =K-8h=\frac{\left( {\sqrt a -\sqrt b }
\right)^{4}}{\sqrt {ab} },\notag\\
& h_{2} =V_{4} =\Psi +32h-6K=\frac{\left( {\sqrt a -\sqrt b }
\right)^{6}}{ab},\notag\\
& h_{3} =U_{2} =2\left( {F+14K-4\Psi -64h} \right)=\frac{\left( {\sqrt a
-\sqrt b } \right)^{8}}{\left( {ab} \right)^{3/2}}\notag
\intertext{and}
& h_{4} =U_{13} =L+44\Psi -120K+512h-20F=\frac{\left( {\sqrt a -\sqrt b }
\right)^{10}}{\left( {ab} \right)^{2}}.\notag
\end{align}

The measure (\ref{eq56}) give generalized Hellingar's discrimination. Let us prove
now its convexity. We can write $h_{t} (a,b)=b\,f_{h_{t} } (a/b)$, $t\in
{\rm N}$, where
\[
f_{h_{t} } (x)=\frac{\left( {\sqrt x -1} \right)^{2(t+1)}}{\left( {\sqrt x }
\right)^{t}}.
\]

The second order derivative of the function $f_{h_{t} } (x)$ is given by
\[
{f}''_{h_{t} } (x)=\frac{\left( {\sqrt x -1} \right)^{2t}}{4\left( {\sqrt x
} \right)^{t+5}}\times A_{5} (x,t),
\]
where
\[
A_{5} (x,t)=t\left( {t+2} \right)\sqrt x \left( {x+1} \right)+2\left(
{t^{2}+t+1} \right)x.
\]

\bigskip
For all $t\ge 0$, $x>0$, $x\ne 1$, we have ${f}''_{h_{t} } (x)>0$. Also
$f_{h_{t} } (1)=0$. In view of Lemma 1.1, the measure $h_{t} (a,b)$ is
convex for all $(a,b)\in {\rm R}_{+}^{2} $, $t\in {\rm N}$.

\bigskip
Following similar lines of (\ref{eq51}) and (\ref{eq52}), the exponential
representation of the measure $h_{t} (a,b)$ is given by
\[
E_{h} (a,b)=\left( {\sqrt a -\sqrt b } \right)^{2}\exp \left( {\frac{\left(
{\sqrt a -\sqrt b } \right)^{2}}{\sqrt {ab} }} \right),\quad (a,b)\in {\rm
R}_{+}^{2} .
\]

\subsection{New Measure}

For all $P,\;Q\in \Gamma_{n} $, let consider the following measures
\begin{equation}
\label{eq57}
M_{t} (a,b)=\frac{\left( {\sqrt a -\sqrt b } \right)^{2(t+2)}}{\left( {a+b}
\right)\left( {\sqrt {ab} } \right)^{t}},\quad t=0,1,2,3,...
\end{equation}

In particular, we have
\begin{align}
& M_{0} =\frac{7}{2}D_{W_{2} W_{1} }^{1} =12D_{CN} -7\Delta =\frac{\left(
{\sqrt a -\sqrt b } \right)^{4}}{\left( {a+b} \right)},\notag\\
& M_{1} =V_{1} =K+26\Delta -48D_{CN} =\frac{\left( {\sqrt a -\sqrt b }
\right)^{6}}{\left( {a+b} \right)\sqrt {ab} },\notag\\
& M_{2} =U_{1} =\Psi +192D_{CN} -100\Delta -8K=\frac{\left( {\sqrt a -\sqrt b
} \right)^{8}}{ab\left( {a+b} \right)},\notag\\
& M_{3} =U_{10} =2\left( {F+22K+4\Delta -5\Psi -128h} \right)=\frac{\left(
{\sqrt a -\sqrt b } \right)^{10}}{\left( {a+b} \right)\left( {ab}
\right)^{3/2}} \notag
\intertext{and}
& M_{4} =U_{15} =\frac{1}{7}\left( {\begin{array}{l}
 7L+448\Delta -1456K+3728\Delta + \\
 +9728h-7680D_{CN} -168F \\
 \end{array}} \right)
=\frac{\left( {\sqrt a -\sqrt b } \right)^{12}}{\left( {a+b} \right)\left(
{ab} \right)^{2}}.\notag
\end{align}

Let us prove now the convexity of the measure (3.17). We can write $M_{t}
(a,b)=b\,f_{M_{t} } (a/b)$, $t\in {\rm N}$, where
\[
f_{M_{t} } (x)=\frac{\left( {\sqrt x -1} \right)^{2(t+2)}}{\left( {x+1}
\right)\left( {\sqrt x } \right)^{t}}.
\]

The second order derivative of the function $f_{M_{t} } (x)$ is given by
\[
{f}''_{M_{t} } (x)=\frac{\left( {x-1} \right)^{2t+2}}{4\left( {x+1}
\right)^{3}\left( {\sqrt x } \right)^{t+5}}\times A_{6} (x,t),
\]
where
\[
A_{6} (x,t)=\left( {\begin{array}{l}
 2\left( {t^{2}+3t+2} \right)x\left( {x^{2}+1} \right)+4\left( {t^{2}+3t+6}
\right)x^{2}+ \\
 +t\left( {t+2} \right)\sqrt x \left( {x^{3}+1} \right)+\left(
{3t^{2}+14t+8} \right)x^{3/2}\left( {x+1} \right) \\
 \end{array}} \right).
\]

\bigskip
For all $t\ge 0$, $x>0$, $x\ne 1$, we have ${f}''_{M_{t} } (x)>0$. Also we
have $f_{M_{t} } (1)=0$. In view of Lemma 1.1, the measure $M_{t} (a,b)$ is
convex for all $(a,b)\in {\rm R}_{+}^{2} $, $t\in {\rm N}$.

\bigskip
Following similar lines of (\ref{eq51}) and (\ref{eq52}), the exponential
representation of the measure $M_{t} (a,b)$ is given by
\[
E_{M} (a,b)=\frac{\left( {a-b} \right)^{4}}{a+b}\exp \left( {\frac{\left(
{\sqrt a -\sqrt b } \right)^{2}}{\sqrt {ab} }} \right),\quad (a,b)\in {\rm
R}_{+}^{2} .
\]

\begin{remark}
\begin{itemize}
\item[(i)] The first ten measures appearing in the second pyramid represents the same measure (\ref{eq11a}) and is the same as $M_{0}$. The last measure given by (\ref{eq49}) is the same as $M_{4}$. The measure (\ref{eq49}) is the only that appears in all the four parts of the last Theorem 2.4. Both these measures generates an interesting measure (\ref{eq57}).
\item[(ii)]  The measure $K_{1}^{1}$ appears in the work of Dragomir et al. \cite{dsb}. An improvement over this work can be seen in Taneja \cite{tan1}.
\item[(iii)] Following similar lines of (\ref{eq51}) and (\ref{eq52}), the
\textit{exponential representation} of the principal measure $L_{t} (a,b)$ appearing in (\ref{eq5}) is given by
\begin{equation}
\label{eq58}
E_{\Delta } (a,b)=\frac{2\left( {a-b} \right)^{2}}{a+b}\exp \left(
{\frac{a+b}{2\sqrt {ab} }} \right),\, (a,b)\in {\rm R}_{+}^{2} .
\end{equation}
We observe that the expression (\ref{eq58}) is little different from the one
obtained above in six parts. Applications of the generating measures (\ref{eq5}),
(\ref{eq49}), (\ref{eq53}), (\ref{eq54}), (\ref{eq55}), (\ref{eq56}) and (\ref{eq57}) along with their exponential representations shall be dealt elsewhere.
\end{itemize}
\end{remark}

\end{document}